\documentclass[%
reprint,
superscriptaddress,
aps,
pra,
]{revtex4-1}

\usepackage[T1]{fontenc} 
\usepackage{graphicx}
\usepackage{dcolumn}
\usepackage{bm}
\usepackage{amsmath,mathtools,amssymb,amsthm,physics}
\usepackage{lipsum}
\usepackage{color,colortbl}
\usepackage{hhline}
\usepackage[normalem]{ulem}
\usepackage{simplewick}
\usepackage{qcircuit}
\usepackage{braket}
\usepackage{float}
\usepackage{multirow}
\usepackage{tikz}
\usepackage{xr-hyper}
\usepackage[colorlinks=true,urlcolor=blue,citecolor=blue,linkcolor=blue]{hyperref}
\usepackage{graphicx,float,wrapfig,import}
\usepackage{mathtools}
\usepackage{enumerate}
\usepackage[ruled, vlined]{algorithm2e}
\usepackage{booktabs}
\usepackage{multirow}
\newtheorem{theorem}{Theorem}[section]
\newtheorem{lemma}[theorem]{Lemma}
\newtheorem{proposition}{Proposition}[section]

\theoremstyle{definition}
\newtheorem{definition}{Definition}[section]

\definecolor{darkcyan}{rgb}{0.0, 0.55, 0.55}

\makeatletter

\def\CT@@do@color{%
	\global\let\CT@do@color\relax
	\@tempdima\wd\z@
	\advance\@tempdima\@tempdimb
	\advance\@tempdima\@tempdimc
	\advance\@tempdimb\tabcolsep
	\advance\@tempdimc\tabcolsep
	\advance\@tempdima2\tabcolsep
	\kern-\@tempdimb
	\leaders\vrule
	\hskip\@tempdima\@plus  1fill
	\kern-\@tempdimc
	\hskip-\wd\z@ \@plus -1fill }
	
\def\bstctlcite{\@ifnextchar[{\@bstctlcite}{\@bstctlcite[@auxout]}}
\def\@bstctlcite[#1]#2{%
 \@bsphack
 \@for\@citeb:=#2\do{%
 \edef\@citeb{\expandafter\@firstofone\@citeb}%
 \if@filesw\immediate\write\csname #1\endcsname{\string\citation{\@citeb}}\fi}%
 \@esphack}
\makeatother

\begin{document}
    \bstctlcite{IEEEexample:BSTcontrol}	

	\title{Enhancing the Quantum Linear Systems Algorithm using Richardson Extrapolation}

	\author{Almudena Carrera Vazquez}
	\affiliation{IBM Quantum, IBM Research -- Zurich}
	\affiliation{ETH Zurich}

	\author{Ralf Hiptmair}
	\affiliation{ETH Zurich}
	
	\author{Stefan Woerner}
	\email{wor@zurich.ibm.com}
	\affiliation{IBM Quantum, IBM Research -- Zurich}
	
	\date{\today}
	
	\begin{abstract}
		We present a quantum algorithm to solve systems of linear equations of the form $A\bm{x}=\bm{b}$, where $A$ is a tridiagonal Toeplitz matrix and $\bm{b}$ results from discretizing an analytic function, with a circuit complexity of $poly(\log(\kappa), 1/\sqrt{\epsilon}, \log(N))$, where $N$ denotes the number of equations, $\epsilon$ is the accuracy, and $\kappa$ the condition number.
  The \emph{repeat-until-success} algorithm has to be run $\mathcal{O}\left(\kappa/(1-\epsilon)\right)$ times to succeed, leveraging amplitude amplification, and sampled $\mathcal{O}(1/\epsilon^2)$ times.
  Thus, the algorithm achieves an exponential improvement with respect to $N$ over classical methods. 
  In particular, we present efficient oracles for state preparation, Hamiltonian simulation and a set of observables together with the corresponding error and complexity analyses.
  As the main result of this work, we show how to use Richardson extrapolation to enhance Hamiltonian simulation, resulting in an implementation of Quantum Phase Estimation (QPE) within the algorithm with $1/\sqrt{\epsilon}$ circuit complexity instead of $1/\epsilon$ and which can be parallelized.
  Furthermore, we analyze necessary conditions for the overall algorithm to achieve an exponential speedup compared to classical methods.
  Our approach is not limited to the considered setting and can be applied to more general problems where Hamiltonian simulation is approximated via product formulae, although, our theoretical results would need to be extended accordingly.
  All the procedures presented are implemented with Qiskit and tested for small systems using classical simulation as well as using real quantum devices available through the IBM Quantum Experience.
	\end{abstract}
	
	\maketitle
	
	\section{Introduction}\label{introsec}

	Systems of linear equations arise naturally in many applications in a wide range of areas.
	In particular, solving sparse tridiagonal systems is at the core of many scientific computation programs \cite{WangParallel}, for instance, for the calibration of financial models, for fluid simulation or for numerical field calculation, where the size of the corresponding systems is often very large \cite{EGLOFF2012309,CHEN2017233,KASPER2001193}.
	In such cases, quantum computing may lead to computational advantages.
	In this work, we analyze and enhance existing quantum algorithms to approximate solution properties for a class of tridiagonal systems that can potentially achieve an exponential speedup over classical approaches.
	
	Approximating the solution to an $s$-sparse positive-definite system with $N$ variables using a classical computer requires, in general, a runtime scaling as $\mathcal{O}(Ns\sqrt{\kappa}\log(1/\epsilon))$ using the conjugate gradient method \cite{ConjGrad}, or, for the case of tridiagonal systems, a runtime scaling as $\mathcal{O}(\sqrt{N})$  when computing in parallel and partitioning the system into blocks of tridiagonal matrices \cite{WangParallel}.
	Here, $\kappa \geq 1$ denotes the condition number of the system, $s$ is the maximum number of nonzero entries per row or column of the corresponding matrix, and $\epsilon > 0$ denotes the targeted accuracy of the approximation. 
	
	In 2009, Harrow, Hassidim, and Lloyd (HHL) proposed a quantum algorithm to estimate properties of solutions requiring a runtime of $\mathcal{O}(\log(N)s^{2}\kappa^{2}/\epsilon)$ \cite{Harrow2009} under the assumptions of efficient quantum oracles for loading the data, Hamiltonian simulation, and estimating the desired property of the solution.
	Thus, this may lead to an exponential speedup with respect to the size of the system $N$ if efficient implementations for these oracles can be found for a problem of interest.
	
	Within this paper, we focus on linear systems given by tridiagonal Toeplitz matrices and right-hand-sides resulting from discretizing smooth functions.
	Under certain assumptions, we provide efficient implementations for each of the before-mentioned oracles as well as a complete error and complexity analysis depending on the problem parameters, and analyze some observables that can be estimated efficiently.
	As the principal result of this paper, we show how to use Richardson extrapolation to enhance Hamiltonian simulation, and we prove that this allows to reduce the circuit complexity, i.e., the total number of gates, of Quantum Phase Estimation (QPE) within the HHL algorithm to $1/\sqrt{\epsilon}$.

	In particular, we obtain a quantum algorithm to solve certain systems of linear equations of size $N$ with a circuit complexity of $polylog(\kappa, 1/\epsilon, N)$.
	The algorithm is a \emph{repeat-until-success} algorithm, and, on expectation, has to be run $\mathcal{O}\left(\kappa/(1-\epsilon)\right)$ times to succeed, leveraging amplitude amplification \cite{brassard2000quantum}.
	Furthermore, since the result is estimated via sampling, we require $\mathcal{O}(1/\epsilon^2)$ successful samples to achieve the target accuracy $\epsilon$.
	Thus, under certain assumptions, this can lead to an exponential improvement in the size of the system over classical methods.
	We combine all our results and provide a concise overview of all assumptions that need to be made to realize this quantum advantage.
	The developed algorithms have been implemented with Qiskit \cite{Qiskit} and tested for small problem instances using a quantum simulator as well as a real quantum device provided by IBM Quantum.

	Note that an overall complexity scaling as $polylog(1/\epsilon)$ has been achieved in \cite{childexp} by approximating the inverse of a matrix via unitaries, which allowed the authors to circumvent the use of QPE.
	The exponential improvement in the complexity was also achieved for the case of the Poisson equation \cite{Cao} without the need to avoid QPE for matrices that can be diagonalized by the sine transform. 
	Finally, as pointed out recently in \cite{fast_pred_wieb}, the cost of the randomization method inspired by adiabatic quantum computation developed in \cite{soma_adiabatic} could be reduced to also run in $polylog(1/\epsilon)$.
	In contrast, our contribution is more general, and, in principle, applicable to all matrices for which Hamiltonian simulation can be approximated via product formulae.
	Furthermore, our approach can be easily parallelized and uses shallow circuits, allowing it to run in near term quantum devices.
	
	Other related work on solving systems of linear equations with a quantum computer which should be highlighted here comprises \cite{Preconditioned, circ_precond, fast_pred_wieb}, addressing the use of a preconditioner, and \cite{DenseMatrices}, tackling the case of dense matrices, although both are not directly relevant in our setting.
	
	Inspired by quantum algorithms \cite{Gilyen2018, tang_2018, tang_2019} give classical algorithms for low-rank matrices in $\mathbb{R}^{m\times n}$ running in $\mathcal{O}(polylog(m, n))$. However, the matrices we consider in this paper are full rank and it was shown that for matrices with larger rank and condition number the quantum-inspired algorithm does not perform well \cite{arrazola}.

	Tridiagonal matrices often result from discretizing one-dimensional differential equations, such as the Poisson equation, which represents a special case of the class we consider here.
	In Appendix~\ref{poissonapp}, we apply the results provided in \cite{Cao} to show how our algorithm can be generalized to solve the Poisson equation also in higher dimensions.
	
	The remainder of the paper is structured as follows. 
	Sec.~\ref{hhlsec} provides a description of the original HHL algorithm, Sec.~\ref{statsect} describes and analyzes procedures to efficiently prepare quantum states with amplitudes specified by a function, and Sec.~\ref{hamsect} introduces an approach for Hamiltonian simulation for tridiagonal Toeplitz matrices. 
	Then, Sec.~\ref{richsect} introduces Richardson extrapolation to reduce the circuit complexity required for Hamiltonian simulation and provides a thorough error and complexity analysis. 
	A procedure to compute the inverse of the eigenvalues as well as the corresponding error scaling are shown in Sec.~\ref{condsect}, and in Sec.~\ref{obsersec}, we explain how to efficiently estimate a set of observables.
	Finally, the overall complexity and error analyses are given in Sec.~\ref{overanalysis}, and the simulated and real hardware results are shown in Sec.~\ref{ressect}.
	
	\textbf{Notation.} Throughout this paper $\norm{\cdot}$ is used to refer to $\norm{\cdot}_2$, the matrix and vector $2$-norm, and $\log(\cdot)$ refers to $\log_2(\cdot)$, while $\ln(\cdot)$ denotes the natural logarithm.
	
	\section{HHL Algorithm}\label{hhlsec}
	
	Let $A \in\mathbb{R}^{N\times N}, N=2^{n_{b}}, n_{b}\in\mathbb{N}$, be Hermitian and $\bm{x},\bm{b}\in\mathbb{R}^{N}$. 
	The problem considered in this paper can be described as: given $A$ and $\bm{b}$, find $\bm{x}$ satisfying $A\bm{x}=\bm{b}$. 
	By rescaling the system, we can assume $\bm{b}$ to be normalized and map it to the quantum state $\ket{b}$, and denote by $\ket{x}$ the normalized solution to the resulting system. 
	In this paper, the mapping used is such that the $i^{\text{th}}$ component of the vector $\bm{b}$ (resp.~$\bm{x}$) corresponds to the amplitude of the $i^{\text{th}}$ computational basis state of the quantum state $\ket{b}$ (resp.~$\ket{x}$). 
	From now on, we will mainly focus on the rescaled system with a normalized solution
	\begin{equation}\label{sytemdesc} 
	A\ket{x}=\ket{b},
	\end{equation}
	and we will use the vector notation to refer to the original problem.
	
	After running the HHL algorithm, the entries of the vector solution to the system, $\ket{x}$, come encoded as the amplitudes of the final state of a quantum register. 
	Therefore one does not have direct access to them, only to a function of the solution vector. 
	Nevertheless, this allows to extract many properties of the solution vector as we discuss in Sec.~\ref{obsersec}.
	
	We proceed now to give an outline of the HHL algorithm. For simplicity, all computations in the following description are assumed to be exact. In particular, in step~(\ref{Step:iii}) we assume that the eigenvalues can be represented exactly with $n_{l}$ binary bits. 
	All quantum registers start in the $\ket{0}$ state at the beginning of the algorithm. 
	\newline
	\newline
	\begin{enumerate}[(i)]
		\item Load the data $\ket{b}\in\mathbb{ R }^{N}$, i.e., perform the transformation
		\begin{equation}\label{initstate}
		\ket{0}_{n_{b}}\mapsto\ket{b}_{n_{b}}=\sum_{i=0}^{N-1}b_{i}\ket{i}_{n_{b}}.
		\end{equation}
		\item Add a $n_l$-qubit register and apply Quantum Phase Estimation (QPE) with $U = e ^ { \mathrm{i} A t } $ and $t$ as specified later in Sec.~\ref{hamsect}. The resulting state is given by
		\begin{equation}
		\sum_{j=0}^{N-1} \beta _ { j } \ket{\lambda _ {j }}_{n_{l}} \ket{u_{j}}_{n_{b}},
		\end{equation}
		where $\ket{u_{j}}_{n_b}$ is the $j^{\text{th}}$ eigenvector of $A$ with respective eigenvalue $\lambda_{j}$, $\ket{\lambda _ {j }}_{n_{l}}$ is the $n_{l}$-bit binary representation of $2^{n_l}\lambda _ {j }$, and $\beta_j$ denotes the components of $\ket{b}$ in the eigenbasis $\ket{u_j}_{n_b}$. 
		\item Add an ancilla qubit and apply a rotation conditioned on $\ket{\lambda_{ j }}_{n_l}$ to get \label{Step:iii}
		\begin{equation}
		\sum_{j=0}^{N-1} \beta _ { j } \ket{\lambda _ { j }}_{n_{l}}\ket{u_{j}}_{n_{b}} \left( \sqrt { 1 - \frac { C^{2}  } { \lambda _ { j } ^ { 2 } } } \ket{0} + \frac { C } { \lambda _ { j } } \ket{1} \right),
		\end{equation}
		where $C$ is a normalization constant.
		\item Apply the inverse QPE. Ignoring possible errors from QPE, this results in
		\begin{equation}\label{eq5}
		\sum_{j=0}^{N-1} \beta _ { j } \ket{0}_{n_{l}}\ket{u_{j}}_{n_{b}}\left( \sqrt { 1 - \frac {C^{2}  } { \lambda _ { j } ^ { 2 } } } \ket{0} + \frac { C } { \lambda _ { j } } \ket{1} \right) .
		\end{equation}
		\item Measure the ancilla
		qubit in the computational basis. If the outcome is $\ket{1}$, the register is in the post-measurement state
		\begin{equation}\label{finalqs}
		\left( \sqrt { \frac { 1 } { \sum_{j=0}^{N-1} \left| \beta_ { j } \right| ^ { 2 } / \left| \lambda _ { j } \right| ^ { 2 } } } \right) \sum _{j=0}^{N-1} \frac{\beta _ { j }}{\lambda _ { j }} \ket{0}_{n_{l}}\ket{u_{j}}_{n_{b}},
		\end{equation}
		which -- up to a normalization factor -- corresponds to the solution $\ket{x}$.
		If the outcome is $\ket{0}$, then repeat steps (i) to (v).
	\end{enumerate}
	
	In the following, we first analyze the errors introduced by the individual steps separately, while considering the other steps exact, and then analyze the overall resulting error and complexity.
	In each case we will denote by $\ket{x}\coloneqq \frac{\bm{x}}{\norm{\bm{x}}}$ the exact normalized solution and by $\ket{\tilde{x}}$ the quantum state returned by the algorithm.

	\section{State preparation}\label{statsect}
	
	In this section we consider the problem of obtaining the quantum state $\ket{b}_{n_b}$ by applying a set of gates to the initial state $\ket{0}_{n_b}$.
	There is plenty of literature addressing state preparation, since it is required by numerous quantum algorithms.
	
	The most generic works discuss how to obtain an arbitrary state, however, with an exponential complexity in the number of qubits, $n_b$, which would diminish any potential quantum advantage obtained by HHL
	\cite{Long2001,Ventura1999,Andrecut,Vartiainen2008,Plesch2011,Shende:2005:SQL:1120725.1120847}. 
	
	Algorithms requiring sub-exponential resources have been achieved for some specific classes of quantum states or by approximation.
	For example, quantum states defined by efficiently integrable probability distributions, e.g., log-concave distributions, can be prepared with polynomial complexity in $n_b$ \cite{Grover2002}.
	Another approach for approximate state preparation is given in \cite{holmes2020efficient}, based on matrix product state approximations.
	This seems a promising direction, although further analysis of their performance will be required.
	Alternatively, \cite{Zoufal2019} recently showed how quantum Generative Adversarial Networks (qGANs) can be used to learn and approximately load arbitrary probability distributions, also requiring only polynomial resources.
	However, this is a heuristic, i.e., it comes without guarantees on the accuracy.
	
	Throughout this paper, we mostly assume $\bm{b}$ is defined by a given function as discussed and analyzed in the following.
	Suppose an analytic function $f:\left[0,1\right] \rightarrow \mathbb{R}$, and suppose $(\bm{b})_i = f\left(x_i\right)$, where $x_i = i/(N-1)$, $i = 0, \ldots, N-1$. 
	To prepare the state, we define $\norm{\bm{b}}_\infty\coloneqq\max_i \abs{(\bm{b})_i}$ and encode the normalized values $f\left(x_i\right) / \norm{\bm{b}}_\infty$ in the $\ket{1}$-amplitude of an ancilla qubit \cite{Zhikuan}, i.e., we assume an operator $F$ defined by
	\begin{equation}
	F: \ket{i}_{n_b}\ket{0}\mapsto\sqrt{1-\frac{f^2(x_i)}{\norm{\bm{b}}^2_\infty}}\ket{i}_{n_b}\ket{0}+\frac{f(x_i)}{\norm{\bm{b}}_\infty}\ket{i}_{n_b}\ket{1}.
	\end{equation}
	Then, the state preparation procedure is given by:
	\begin{enumerate}[(i)]
		\item Prepare an $n_b$-qubit equal-superposition state and append an ancilla qubit, obtaining
		\begin{equation}
		\ket{\psi}=\frac{1}{\sqrt{N}}\sum_{ i = 0 }^{N-1}\ket{i}_{n_{b}}\ket{0}.
		\label{eq:state_prep_equal_superposition}
		\end{equation}
		\item Apply $F$ to get 
		\begin{equation}
		\frac{1}{\sqrt{N}}\sum_{ i = 0 }^{N-1}\ket{i}_{n_{b}}\left(\sqrt{1-\frac{f^2(x_i)}{\norm{\bm{b}}^2_\infty}}\ket{0}+\frac{f(x_i)}{\norm{\bm{b}}_\infty}\ket{1}\right).
		\end{equation}
		\item Measure the ancilla qubit. If observing $\ket{1}$, the resulting state is $\ket{b}_{n_b} \ket{1}$. If the outcome is $\ket{0}$, then repeat steps (i) to (iii).
	\end{enumerate} 
	
	$F$ can be realized by first using quantum arithmetic to compute $\arcsin(f(x))$ into an ancilla register, and then using controlled Y-rotations $R_y$ to map this to the amplitude of the ancilla qubit.
	The complexity of this approach depends on $f$ but often will be dominated by the computation of the arcsine, which can be realized using $\mathcal{O}(n_b^2)$ operations \cite{Hner2018OptimizingQC}.
	The constant overhead of quantum arithmetic can be quite large, likely rendering this approach impractical for the near future \cite{Hner2018OptimizingQC}.
	Thus, in the following, we show and analyze how $F$ can be efficiently approximated.
	
	In \cite{Woerner2018} it is shown how to implement an operator
	\begin{equation}\label{Pequation}
	P_p: \ket{i}_{n_{b}}\ket{0}\mapsto \ket{i}_{n_{b}}\left(\cos(p(x_i))\ket{0}+\sin(p(x_i))\ket{1}\right),
	\end{equation}
	where $p:\left[0,1\right] \rightarrow \mathbb{R}$ is a polynomial of degree $d$.
	The corresponding quantum circuits, illustrated in Fig.~\ref{stateprepcirc} for $d=2$, use polynomially many (multi-controlled) $R_y$ gates, as stated in the following lemma.
	The proofs of all lemmas in this section can be found in Appendix~\ref{Appendix:StatePreparation}.
	
	\begin{lemma} \label{lemma:poly_circuit}
		Let $P_p$ be the operator described in Eq.~\ref{Pequation} associated to a polynomial $p$ of degree $d$ on $n$ qubits, and suppose $d\leq \lceil \frac{n}{2}\rceil$. Then implementing $P_p$ requires $\mathcal{ O }\left(n^{d}\right)$ CNOT and single-qubit gates.
	\end{lemma} 
	
	\begin{figure}[ht!]
		\includegraphics[width=\linewidth]{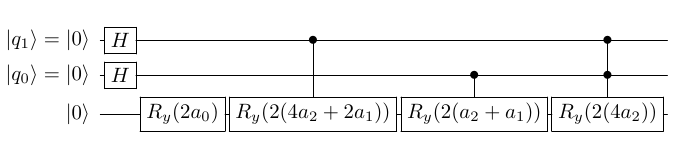}
		\caption{Circuit preparing a state with amplitudes given by the polynomial $p(x)=a_{2}x^{2}+a_{1}x+a_{0}=(4a_{2}+2a_{1})q_{1}+(a_{2}+a_{1})q_{0}+4a_{2}q_{1}q_{0}+a_{0}$, for $x = 0, 1, 2, 3$, represented by two qubits.}
		\label{stateprepcirc}
	\end{figure}
	
	Similarly as in \cite{Woerner2018, OptionPricing, qsimopt}, we can exploit the local linearity of the sine around $x=0$ and consider the scaled function $c p(x)$, for a small $c\in (0,1]$, i.e., the amplitudes of $\ket{i}_{n_b}\ket{1}$ in Eq.~\ref{Pequation} then equal $\sin(cp(x_i)) = c p(x_i) + \mathcal{O}(c^{3})$ and we can drop the arcsine, while introducing only a small error.
	
	For every analytic function $f$ we can find a sequence of polynomials with linearly increasing degree that converges exponentially fast to $f$ \cite{polyapprox}.
	Given a suitable polynomial approximation $p_f$ of $f/\norm{\bm{b}}_\infty$ and a small rescaling factor $c$, we can efficiently construct
	\begin{equation}
	F_{c} \ket{i}_{n_{b}}\ket{0}\approx \ket{i}_{n_{b}}\left(\sqrt{1-(cp_f(x_i))^{2}}\ket{0}+cp_f(x_i)\ket{1}\right).
	\end{equation}
	
	This approach only makes sense if the probability of success, i.e., measuring $\ket{1}$, is not decaying too fast with increasing $N$ and decreasing $c$.
	Let $I_{n_b}$ denote the $N \times N$ dimensional identity matrix and let $\norm{f}_{2}$ denote the $L^{2}_ {[0, 1]}$-norm of $f$.
	The following lemma provides a lower bound on the success probability of this procedure, which is then discussed in the remainder of this section.
	
	\begin{lemma}\label{statexp}
		Suppose $M_{1}\coloneqq I_{n_{b}} \otimes \ket{1}\bra{1}$, an analytic function $f:\left[0,1\right] \rightarrow \mathbb{R}$, a polynomial approximation $p_f$ such that $\norm{f/\norm{\bm{b}}_{\infty} - p_f}_{L^\infty([0,1])}=\epsilon_p$, and $c\in (0,1]$. Then,
		\begin{equation}
		\bra{\psi}F_{c}^{\dagger}M_{1}^{\dagger}M_{1}F_{c}\ket{\psi} \geq  \frac{c^2\norm{\bm{b}}^2}{N\norm{\bm{b}}^2_\infty}-c^2\epsilon_p
		\end{equation}
		for $c\rightarrow 0$, where $\ket{\psi}$ is given in Eq.~\ref{eq:state_prep_equal_superposition}. 
	\end{lemma}
	
	Note that this approach can also be used in combination with quantum arithmetic. For instance, we can compute $f(x)$ into an ancilla register and then use the sine-approximation only to approximate the arcsine. This might lead to overall better circuits and should be considered automatically in a future quantum compiler for the HHL algorithm.
	
	Lem.~\ref{lemma:poly_circuit} implies a complexity of
	\begin{equation}\label{Eq:state_prep_complexity}
	    \mathcal{O}\left(\log^d(N)\right),
	\end{equation}
	which combined with Lem.~\ref{statexp} and the assumption that $c$ is small implies an expected complexity of the procedure (i.e., the number gates multiplied by the expected number of repetitions until success) of
	\begin{equation}\label{Eq:state_prep_expected_complexity}
	    \mathcal{O}\left(\log^{d}(N)\sqrt{N}\norm{\bm{b}}_\infty/c\norm{\bm{b}}\right),
	\end{equation} 
	where $d$ can either denote the degree of the approximating polynomial or the complexity of quantum arithmetic, i.e., usually $\leq 3$, and where we leverage amplitude amplification.
	Thus, in order to preserve the exponential speedup of the HHL algorithm, we need to assume that $\norm{\bm{b}}/\norm{\bm{b}}_\infty = \Omega(\sqrt{N})$, which implies an expected complexity of $\mathcal{O}(\log^{d}(N)/c)$.
	In other words, the entries of $\bm{b}$ need to be relatively uniform, as already pointed out in \cite{aaronsonfine}, without a few $\bm{b}_{i}$ being ``vastly larger'' than the others.
	In the following, we refer to this as the \emph{uniform assumption}.

	To conclude, the following lemma shows how the polynomial state approximation introduced in this section affects the accuracy of the solution returned by the HHL algorithm, $\ket{\tilde{x}}$, compared to the analytic solution, $\ket{x}$, assuming all other steps in the algorithm are exact.
	\begin{lemma}\label{psperr}
		Let $c\in(0,1]$ be the rescaling constant used to approximate $\ket{b}$, and $\ket{\tilde{b}}$ the approximated state by $p_f$ such that $\norm{f-p_f}_{L^\infty([0,1])}=\epsilon_p$. Then
		\begin{equation}
		\norm{\ket{x} - \ket{\tilde{x}}} 
		\leq \frac{4\kappa\sqrt{N}\norm{\bm{b}}_\infty\left(\epsilon_p+c^2\right)}{\norm{\bm{b}}}+\mathcal{O}\left( c^4\right)  
		\eqqcolon \epsilon_S
		\end{equation}	
		for $c \rightarrow 0$, where  $\kappa$ is the condition number of $A$.
	\end{lemma}

	If $\bm{b}$ satisfies the \emph{uniform assumption}, then, due to the standard norm inequality $\norm{\bm{b}} \leq \sqrt{N} \norm{\bm{b}}_{\infty}$, we can define $\sqrt{N}\norm{\bm{b}}_\infty / \norm{\bm{b}} \eqqcolon C_b \in \Theta(1)$ and treat it as a constant. 
	Thus, for an overall accuracy of $\epsilon_S$, we can for instance take
	\begin{equation}
	    \epsilon_p =\frac{\epsilon_S}{8\kappa C_b}, \quad\text{and}\quad c = \sqrt {\epsilon_p},
	\end{equation}
	ignoring higher order terms.
	
	As expected, this implies that approximate state preparation is not suitable for ill-conditioned matrices, e.g., with $\kappa = \Omega(N)$, since the success probability would drop quickly in $N$.
	In these cases, we need to apply quantum arithmetic, which allows us to achieve a required accuracy with $\mathcal{O}(\log(\kappa))$ additional qubits to represent the required intermediate results.
	
	\section{Tri-diagonal Hamiltonian simulation}\label{hamsect}
	
	In this section we consider the Hamiltonian simulation problem, i.e., we want to construct a quantum circuit acting as $e^{\mathrm{i}At}$ for $t \in \mathbb{R}$ and a Hermitian matrix $A \in \mathbb{R}^{N \times N}$.
	In particular, we focus on the case when $A$ is a symmetric real Toeplitz matrix with constant main diagonal with value $a\in\mathbb{R}$ and constant off-diagonals with value $b\in\mathbb{R}$, and zeroes elsewhere.
	As an illustration, in the $4\times 4$ case, $A$ takes the form
	\begin{equation}
	A=\begin{pmatrix}
	a & b & 0 & 0 \\
	b & a & b & 0 \\
	0 & b & a & b \\
	0 & 0 & b & a \\
	\end{pmatrix}
	, \quad a , b \in \mathbb { R }.
	\end{equation}
	Matrices like this arise, for instance, when solving the Poisson equation using finite difference approximation on $N+2$ grid points, see Appendix~\ref{poissonapp} for more details.
	
	Hamiltonian simulation arose from trying to simulate the dynamics of quantum systems, the original motivation for quantum computers \cite{Feynman1982}.
	Nowadays, the scope of Hamiltonian simulation has widened to include an array of applications, such Gibbs state preparation \cite{gibbs} or option pricing \cite{baaquie}, just to name a few.
	
	The cases for sparse and dense matrices are usually treated differently.
	We are interested in the former.
	For dense Hamiltonians, see, e.g., \cite{Childs2009,Childs2010}.
	
	For sparse matrices there are four approaches that need to be highlighted.
	Early methods were based on Lie-Trotter product formulae \cite{Berry2007,Childstar,Childs:2003:EAS:780542.780552,Ahokas2004}, which perform well, e.g., if parts of the Hamiltonian commute. 
	Other approaches are based on quantum walks \cite{Childs2010,Berryblack} and are optimal in the sparsity and time of the evolution but have an unfavorable error dependence.
	In contrast, the fractional-query model \cite{Berryfrac} is optimal in the error but not in the sparsity dependence.
	Finally, the linear combination of unitaries is a non-deterministic technique that achieves a better dependence on the error at the cost of a worse scaling in the size of the system, albeit still $polylog(N)$ \cite{ChiWie,Berryop,Berrylcu}. 
	
	Our work is based on Lie-Trotter product formulae.
	Particularly, on the method first introduced in \cite{Aharonov2003} due to its efficiency for the type of matrices we consider, and due to being compatible with Richardson extrapolation, which will allow us to improve the error dependence of the complexity from $\mathcal{O}(1/\epsilon)$ to $\mathcal{O}(\log^3(1/\epsilon))$, as we show in Sec.~\ref{richsect}.
	
	This method consists of three steps:
	\begin{enumerate}[(i)]
		\item Find a decomposition $A= \sum_{j=1}^J H_{j}$, where each $H_{j}$ is a $1$-sparse matrix (at most one nonzero entry in each row or column).
		
		\item Find an efficient implementation for each $e ^ {\mathrm{i}H_{j}t}$.
		
		\item If the matrices $H_{j}$ do not commute, i.e., $e^{\mathrm{i}\sum_{j=1}^J H_{j}}\neq\prod_{j=1}^J e^{\mathrm{i}H_{j}}$, use a Lie-Trotter approximation for $e ^ { \mathrm{i} A t }$.\label{Step:iiiHam}
	\end{enumerate}
	
	Following this approach we achieve a method that is linear in the number of qubits. This is done by taking the decomposition $H_{1}=aI_{n_{b}}$,
	\begin{equation*}
	H_{2}=bI_{n_{b}-1}\otimes \sigma_{x}
	= \begin{pmatrix}
	0 & b & & &  \\
	b & 0 & & &  \\
	 & &\ddots&  & \\
	 & &      &   0 & b \\ 
	 & & & b & 0
	\end{pmatrix}, 
	\end{equation*}
	and
	\begin{equation*}
	H_{3}=\begin{pmatrix}
	0& &  & & & &\\
	&0 & b & & & & \\
	&b & 0 & & & & \\
	& & &\ddots& & & \\
	& & &      &   0 & b &\\ 
	& & & & b & 0&\\
	& &  & & & &0
	\end{pmatrix}.
	\end{equation*}
	The circuits for each case can be found in Appendix~\ref{hamcircapp}.
	
	For step~(\ref{Step:iiiHam}), $H_{1}$ commutes with the other two matrices, thus, we use the a Lie-Trotter formula to approximate $e^{\mathrm{i}(H_{2}+H_{3})}$ as introduced below.
	
	\begin{definition}[Lie-Trotter-Suzuki, \cite{suzukiformula}]\label{def:lie-trotter-suzuki}
    Let $H\in\mathbb{C}^{N\times N}$ be a Hermitian matrix with $H=\sum_{j=1}^{J} H_{j}$, and $\chi\in\mathbb{N}$. Then, the $2\chi$-order formula approximating $e^{\mathrm{i}Ht}$ is recursively defined by
    \begin{eqnarray}
        S_{1}(t) &=& \prod_{j=1}^{J}e^{-\mathrm{i}H_j t/2}\prod_{j=J}^{1}e^{-\mathrm{i}H_j t/2},\\
        S_{\chi}(t) &=& 
        \left(S_{\chi-1}\left(s_{\chi-1}t\right)\right)^{2} \times \nonumber \\
        && S_{\chi-1}\left(\left(1-4s_{\chi-1}\right) t \right) \times \\
        && \left(S_{\chi-1}\left(s_{\chi-1}t\right)\right)^{2}, \nonumber
    \end{eqnarray}
    where $s_{p}=\left(4-4^{1/(2p+1)}\right)^{-1}$ for $p\in\mathbb{N}$.
    \end{definition}
	
	In this paper, we use the second order Lie-Trotter formula, i.e., $S_1(t)$, and approximate Hamiltonian simulation by applying $m$ Trotter steps, leading to
	\begin{equation}\label{noneftrot}
	S_1(t/m)^m = \left( e^{\mathrm{i}H_{2} t / 2m} e^{\mathrm{i}H_{3} t / m} e^{\mathrm{i}H_{2} t / 2m} \right)^{m},
	\end{equation}
	also known as the symmetric Strang splitting \cite{Strangsplit}.
	Rearranging the equation leads to
	\begin{equation}\label{trott2}
	e ^ { - \mathrm{i} H _ { 2 } t / 2 m } \left( e ^ { \mathrm{i} H _ { 2 } t / m } e ^ { \mathrm{i} H _ { 3 } t / m } \right) ^ { m } e ^ { \mathrm{i} H _ { 2 } t / 2 m },
	\end{equation}
	which is more efficient than Eq.~\ref{noneftrot}, since the circuit for $e^{\mathrm{i}H_{2}t}$ is used only $m+2$ instead of $2m$ times.
	Including $e^{\mathrm{i}H_{1}t}$, we get 
	\begin{equation}
	    V(t,m) \coloneqq e^{\mathrm{i}H_{1}t} S_1(t/m)^m. 
	\end{equation}
	Then \cite[Thm.~4.3]{Nielchu} shows that
	\begin{equation}
	\lim_{m\rightarrow \infty} V(t,m) = e^{\mathrm{i}At}.
	\end{equation}
	
	The HHL algorithm uses Hamiltonian simulation within QPE, i.e., it uses powers of $e^{\mathrm{i}At}$.
	The following lemma provides a formula for the approximation error of Hamiltonian simulation in terms of the power $k$, the number of Trotter steps $m$, the time $t$, and the off-diagonal coefficient $b$. 
	All lemmas in this section are proved in Appendix~\ref{Appendix:Hamiltonian}.
	
	\begin{lemma}\label{hamerreq}
		Let $k, m \in \mathbb{N}$, and $t \geq 0$. Then,
		\begin{equation}
		\norm{\left(e^{\mathrm{i}At}\right)^{k}-V(t,m)^{k}} \leq \frac{kt^{3}b^{3}}{2m^{2}} + \mathcal{O}\left(\frac{t^4}{m^4}\right) \quad \text{for } \frac{t}{m} \rightarrow 0.
		\end{equation}
	\end{lemma}
	
	Therefore, for a target error $\epsilon_{A}$, it is enough to choose
	\begin{equation}\label{mchoice}
	m(k,t) = \mathcal{O}\left(\sqrt{\frac{k t^{3} b^{3}}{2\epsilon_{A}}}\right),
	\end{equation}
	where $m(k,t)$ denotes the number of Trotter steps for approximating $\left(e^{\mathrm{i}At}\right)^{k}=e^{\mathrm{i}Atk}$.
	However, as detailed in Appendix~\ref{richapp}, this implies that the overall complexity from estimating the eigenvalues scales as $\mathcal{O}(1/\epsilon)$ for an HHL target accuracy $\epsilon$.
	However, in Sec.~\ref{richsect} we show how to reduce this complexity to $\mathcal{O}\left(\log^{3}(1/\epsilon)\right)$.
	
	Finally, the following lemma shows how approximating Hamiltonian simulation by the procedure described in this section affects the accuracy of the solution returned by HHL. 
	
	\begin{lemma}\label{hamoverall}
		Let $\ket{\tilde{x}}$ denote the final state of the algorithm assuming simulation of $e^{\mathrm{i}At}$ with accuracy $\epsilon_{A}$ and exact procedures for all other steps. Then	
		\begin{eqnarray}
		\norm{\ket{x} - \ket{\tilde{x}}}
		&<& 2\epsilon_A +\mathcal{O}\left(\epsilon_A^2\right) \quad \text{for } \epsilon_{A} \rightarrow 0,
		\end{eqnarray}
		where $\ket{x}$ denotes the exact solution.
	\end{lemma}
	
	\section{Richardson Extrapolation for Hamiltonian simulation}\label{richsect}
	
	In this section we show how to use a general form of Richardson extrapolation \cite{Richardson, Numrich}, called multi-product decomposition \cite{Chin2010,ChinOperator}, to obtain a quadratic speedup for the QPE within the HHL algorithm.
	More precisely, we show how the $\epsilon$ dependency of QPE can be improved from $\mathcal{O}\left(1/\epsilon\right)$ to $\mathcal{O}\left(1/\sqrt{\epsilon})\right)$.
	In essence, the idea is to approximate $e^{\mathrm{i}At}$ by a linear combination of lower-order product formulas, such as the Lie-Trotter formula.
	
	In the quantum computing context, Richardson extrapolation has been used previously for error mitigation \cite{Temme2017,Kandala2018}.
	The idea of using a multi-product formula to approximate Hamiltonian simulation was introduced in \cite{ChiWie}, and further studied in \cite{lowwie}.
	The main difference between the present manuscript and this previous work is that we propose to compute the linear combination classically, while they keep the whole process within a quantum circuit. In addition, we extend the existing theory and provide a general formula to calculate the optimal number of extrapolation points, and introduce the before-mentioned exponential reduction of gates in the QPE context.
	This scheme of classically combining the results has already shown promise in the numerical experiments performed for Hamiltonian simulation with $2$ and $3$ extrapolation points in \cite{hamextrap}, where they also apply an additional extrapolation to first mitigate the physical errors of each computation.
	This section is mostly expressed in general terms, since the results on Hamiltonian simulation are not limited to the HHL algorithm.
	
	We first run Hamiltonian simulation multiple times with different small numbers of Trotter steps.
	Then, we classically recombine the results to cancel out the lower order error terms and to obtain a better approximation of the considered observable.
	The rationale is that running $l$ instances of the algorithm with small numbers of Trotter steps should be more efficient than one single run with a large number of Trotter steps, while achieving the same accuracy.
	
	Usually, when applying extrapolation schemes, the number of extrapolation points is fixed and error estimates are given in terms of a varying grid or step size.
	Due to the assumption of fixed number of extrapolation points, most known extrapolation error bounds hide their dependency on the number of points.
	In contrast, we vary the number of extrapolation points, which implicitly sets the step size.
	Thus, we need to derive the explicit dependency of the error on the number of points.
	
	In the following, we first give a formal definition of the multi-product formulae, and then, compare the resulting use of resources against running the algorithm for the same accuracy without extrapolation.

    \begin{definition}[Multi-product formula]
    For $l \in \mathbb{N}$, let $m_1,\dots,m_l$ be distinct natural numbers, and $a_1,\dots,a_{l}\in\mathbb{R}$ satisfy $\sum_{j=1}^{l}a_j =1$. Then, the multi-product formula $M_{l,\chi}(t)$ is given by
    \begin{equation}
        M_{l,\chi}(t) = \sum_{j=1}^{l}a_j S_{\chi}(t/m_j)^{m_j},
    \end{equation}
    where $\chi, S_{\chi}$ are as in Def.~\ref{def:lie-trotter-suzuki}.
    \end{definition}
	
	As in Sec.~\ref{hamsect}, we focus on 
	the second-order product formula, which leads to
	\begin{equation}\label{Eq:Vldef}
	    V_l(t,\vec{m}_l) = \sum_{j=1}^{l}a_j  S_1(t/m_j)^{m_j},
	\end{equation}
	where $\vec{m}_l=(m_1,\dots,m_l)$ are the different Trotter exponents used for the extrapolation.
	Furthermore, we consider coefficients given by the closed form \cite{Chin2010}
	\begin{equation}\label{ajdef}
	    a_j = \prod_{\substack{q\in\{1,\dots,l\} \\ q\neq j}}\frac{m_j^{2}}{m_j^2 -m_q^2},\quad j=1,\dots,l.
	\end{equation}
	
	The following theorem studies the approximation error of Hamiltonian simulation using a multi-product formula scheme, and provides an expression to calculate the optimal number of extrapolation points in terms of the target accuracy.
	All proofs from this section can be found in Appendix~\ref{Appendix:richext}.
	
	\begin{theorem}\label{Thm:extraperr}
	Let $H=H_2+H_3$ be a Hermitian matrix such that $\comm{H_2}{H_3}\neq 0$, $d=\max\{\norm{H_2},\norm{H_3}\}$, $t\in \mathbb{R}$ and $V_l(t,\vec{m}_l)$ as defined in Eq.~\ref{Eq:Vldef}. Then
	\begin{equation}
	    \norm{e^{\mathrm{i}Ht}-V_l(t,\vec{m}_l)} \leq \frac{\left(2dt\right)^{2l+1}}{(2l+1)!}\prod_{i=1}^l\frac{1}{m_i^2}.
	\end{equation}
	For an accuracy $\epsilon$, it is then enough to set $\vec{m}_l=(1,\dots,l)$ for
	\begin{equation}\label{Eq:l_def}
	    l=\frac{\ln(t/2\sqrt{2\pi}\epsilon)}{4W\left(\frac{\ln(t/2\sqrt{2\pi}\epsilon)}{4e\sqrt{dt}}\right)},
	\end{equation}
	where $W(x)$ denotes the Lambert function (cf.~\cite{Corless1996}).
	\end{theorem}
	
	However, note that we are interested in approximating $e^{\mathrm{i}Atk}$ for $k=0,\dots,2^{n_l}-1$ (recall that QPE makes calls to only powers of 2 of the input unitary). 
	One option would be to repeat the circuits $k$ times as in Sec.~\ref{hamsect}, i.e. implementing $ S_1\left(t/m_j\right)^{km_j}$.
	Although the error resulting from this approach scales better in terms of $k$, we would not be able get rid of a $2^{n_l}$ gate complexity.
	Instead we will implement
	\begin{equation}\label{Eq:Vlkdef}
	    V_l\left(tk,\vec{m}_l(k)\right) = \sum_{j=1}^{l}a_j  S_1\left(tk/m_j(k)\right)^{m_j(k)},
	\end{equation}
	where $m_j(k) = \left(\lfloor k^{\frac{3}{n_l}}\rfloor + 1\right) m_j$ for $\vec{m}_l(k) = (m_1(k),\dots,m_l(k))$, and take as the number of extrapolation points the worst case, i.e. substituting in Eq.~\ref{Eq:l_def} for $tk = t2^{n_l}$.

	Later in Sec.~\ref{overanalysis} we will see that $dt$ in Eq.~\ref{Eq:l_def} is bounded. 
	Then, as shown in Appendix~\ref{Appendix:richext}, we can take $l=1/\sqrt{\epsilon}$.
	Therefore, as detailed in Appendix~\ref{richapp}, the overall complexity of estimating the eigenvalues in the extrapolation scheme scales as
	\begin{equation}\label{Eq:qpe_parallel_complexity}
	\mathcal{O}(n_l\log(N)/\sqrt{\epsilon})
	\end{equation}
	if the $l$ HHL algorithms can be run in parallel.
	
	Finally, from Lem.~\ref{hamoverall} and Thm.~\ref{Thm:extraperr}, the main theorem of this section follows, which studies the accuracy of extrapolation with the solutions returned by $l$ HHL algorithms where only Hamiltonian simulation is approximated.
	\begin{theorem}\label{thm:hhl_extrapolation}
	Let $a_j$ and $\vec{m}_l(k)$ defined as before, and suppose 
	\begin{equation}
	    \norm{e^{\mathrm{i}Ht}-V_l(t,\vec{m}_l)}<\epsilon_A.
	\end{equation}
	Furthermore, let $\ket{\tilde{x}_j}$ denote the solution returned by HHL approximating $e^{\mathrm{i}Atk}$ with
	$V\left(tk,m_j(k)\right) = e^{\mathrm{i}H_1tk}S_1(tk/m_j(k))^{m_j(k)}$, and let $\ket{\tilde{x}}$ denote the solution obtained by combining the results of 
	 $l$ independent algorithms, i.e.
	\begin{equation}
	    \ket{\tilde{x}}=\sum_{j=1}^{l} a_j \ket{\tilde{x}_j}.
	\end{equation}
	Then, for $\epsilon_{A}\rightarrow 0$, we have
	\begin{equation}
	    \norm{\ket{x} - \ket{\tilde{x}}} < 2\epsilon_A+\mathcal{O}(\epsilon_A^2).
	\end{equation}
	\end{theorem}

	\section{Eigenvalue inversion}\label{condsect}

	In this section we focus on the conditional rotation step of the HHL algorithm to invert the eigenvalues.
	Combined with QPE, this is the step that actually solves the linear system of equations.

	Recall that for a given $m$-qubit unitary $U$ with eigenvector $\ket{\psi}_{m}$ and eigenvalue $e^{2\pi \mathrm{i}\theta}$ QPE approximates $\theta$ using $n$ bits.
	More precisely, given $U$ and $\ket{0}_{n}\ket{\psi}_{m}$, it returns the state $\ket{\tilde{\theta}}_{n}\ket{\psi}_{m}$ with high probability, where $\tilde{\theta}$ denotes an $n$-bit representation of $2^n\theta$.

	After QPE has been applied within HHL, and assuming no errors in the computations, the system is in the state
	\begin{equation}
	\sum_{j=0}^{N-1}\beta_{j}\ket{\tilde{\lambda}_{j}}_{n_{l}}\ket{u_{j}}_{n_{b}}\ket{0},
	\end{equation}
	where $\tilde {\lambda}_{j}$ is the $n_{l}$-bit representation of $N_l\lambda_{j} t/2\pi$, $N_l\coloneqq 2^{n_l}$, $\lambda_{j}$ denotes the $j^{\text{th}}$ eigenvalue of our matrix $A$ with $\lambda_{\min} \coloneqq \lambda_{0} \leq \lambda_{1} \leq \ldots \leq \lambda_{N-1} \eqqcolon \lambda_{ \max }$, and $t$ is the time of the Hamiltonian simulation as introduced in Sec.~\ref{hamsect} and \ref{richsect}.
	We will choose $t \in (0, 2\pi/\lambda_{\max}]$ and $n_l$ sufficiently large to guarantee that $\tilde \lambda_j \in [a, N_l-1]$, where $a\in\mathbb{Z}$, as will be discussed in more detail in the following and in Sec.~\ref{overanalysis}.
	
	To solve the linear system, the goal is to find a circuit performing the transformation
	\begin{align}\label{transform}
	&\sum_{j=0}^{N-1}\beta_{j}\ket{0}\ket{\tilde{\lambda}_{j}}_{n_{l}}\ket{u_{j}}_{n_{b}}\mapsto \nonumber \\
	&\sum_{j=0}^{N-1}\beta_{j}\left(\sqrt{1-\frac{C^2}{\tilde\lambda_{j}^{2}}}\ket{0}+\frac{C}{\tilde\lambda_{j}}\ket{1}\right)\ket{\tilde{\lambda}_{j}}_{n_{l}}\ket{u_{j}}_{n_{b}},
	\end{align}
	measure the last qubit, and repeat the procedure until the measurement results in a $\ket{1}$.
	To achieve this, we can apply a controlled Y-rotation on the target qubit conditioned on $\ket{\tilde \lambda_j}_{n_l}$. The rotation angles are given by
	\begin{equation}\label{function}
	f \left( \tilde \lambda \right) = \arcsin \left( \frac {C} {\tilde \lambda} \right),
	\end{equation}
	which takes values in $[0, 1]$ due to our assumptions on $C$, $t$ and $n_l$.

    There are multiple ways to implement inverse and arcsine:
    \cite{Cao} discuss how to combine Newton's method to compute the inverse and a bisection search to calculate the arcsine, and show that the run time and number of required additional qubits scale as $polylog(1/\epsilon)$, where $\epsilon > 0$ denotes the resulting $L_{\infty}$-approximation error.
    Although this is efficient asymptotically, the constant overhead of this method is large.
    The authors of \cite{Hner2018OptimizingQC} show how to use piecewise polynomial approximation for arbitrary functions and provide an empirical study of the performance.
    They introduce an algorithm that iteratively constructs intervals and corresponding polynomials such that the resulting piecewise polynomial approximation satisfies a given $L_{\infty}$-approximation error target.
    However, they do not provide a theory on the performance, i.e., they do not analyze the dependence of the approximation error or the number of approximation intervals and degree of the polynomials. 
    
    In the following, we analyze a piecewise polynomial approximation scheme similar to \cite{Hner2018OptimizingQC} and provide a rigorous performance theory to approximate $\arcsin(C/x)$ for $x \in [a, N_l-1]$, where $a\in\mathbb{Z}$ will be specified later, and for $x\in[1,a]$ we apply the identity.
    Note that by construction the algorithm introduced in \cite{Hner2018OptimizingQC} will be at least as efficient as ours, i.e., our analysis implies a rigorous performance bound for this approach as well.
    Thus, in practice, the approach introduced in \cite{Hner2018OptimizingQC}, equipped with our theory, is likely the best choice.
    
    Suppose an exponentially growing set of approximation intervals $[a_i, a_{i+1}]$, for $i=1, \ldots, M$, where $a_1 = a$, $a_{i+1} = 5a_i$, and $M=\lceil\log_5(\frac{N_l-1}{a})\rceil$ such that the whole interval $[a, N_l-1]$ is covered.
    For each interval we use \emph{Chebyshev interpolation} \cite{Deuflhard_2003} to approximate 
    $\arcsin(C/x)$ with a polynomial of degree $d$.
    The following lemma provides the resulting approximation error.
    The proofs of all results in this section can be found in Appendix~\ref{Appendix:conditionalrot}.
    
	\begin{lemma}\label{lemma:chebychev}
		Let $x \in [a, N_l-1]$, $a\in\mathbb{Z}$, and $f(x)=\arcsin(C/x)$, where $C/x\in(0,1]$.
		Let $p_f(x)$ be the piecewise polynomial approximation returned by the introduced scheme for polynomial degree $d$. Then we have
		\begin{equation}
		    \epsilon_C \coloneqq \norm{f-p_f}_{L^{\infty}([a,N_l-1])} \leq \frac{8.13\sqrt{\ln^2(r)+(\pi/2)^2}}{2^{d+1}-1},
		\end{equation}
		where 
		\begin{equation}
		    r\coloneqq \frac{2C}{a}+\sqrt{\abs{1-\left(\frac{2C}{a}\right)^2}},
		\end{equation}
		and $\norm{f}_{L^{\infty}([a,b])}\coloneqq\max_{x\in[a,b]}\abs{f(x)}$.
	\end{lemma}
	
	The complexity of evaluating $M$ piecewise polynomials of degree $d$ on $n_l$ qubits can be achieved in $\mathcal{O}\left(n_l^2 d + M d \log(M)\right)$ Toffoli gates and $(d+1) n_l + \lceil \log(M)\rceil +1$ qubits through a parallel reversible implementation of the classical Horner scheme \cite{Hner2018OptimizingQC, hornerclassical}.
	
	Note that a decomposition into equally-sized intervals of the domain would result either in a large degree of the polynomials, or in a large number of intervals, and, in both cases, would not be efficient.
	
	The following lemma analyzes the impact of the approximation error of the eigenvalue inversion to the error of the whole HHL algorithm which allows us to chose $d$ accordingly.
	Here $\ket{\tilde{x}}$ will denote the normalized state returned by the algorithm assuming exact procedures for everything except for the eigenvalue inversion and $\ket{x}$ denotes the exact normalized solution.
	
	\begin{lemma}\label{condoverall}
		Let $\epsilon_C$ denote the error in approximating $f(x)=\arcsin(C/x)$ for $x\in[a,N_l]$, where $a\in\mathbb{Z}$. 
		Furthermore, let $\epsilon_R>0$, $\lambda_{\min}$ be the smallest eigenvalue of $A$, and $t\in\mathbb{R}$ the Hamiltonian simulation evolution time. 
		If 
		\begin{equation}
		    \epsilon_C = \frac{\epsilon_R}{2(2\kappa^2-\epsilon_R)},\quad a=2^{\frac{2n_l}{3}},\quad C = \frac{N_l t\lambda_{\min}}{2\pi},
		\end{equation}
		 and
		\begin{equation}
		    n_l = 3\left(\left\lfloor \log(\frac{2(2\kappa^2-\epsilon_R)}{\epsilon_R}+1)\right\rfloor+1\right),
		\end{equation}
		then
		\begin{equation}
		\norm{\ket{x}-\ket{\tilde{x}}}\leq\epsilon_R.
		\end{equation}
		Moreover, the probability of successfully inverting the eigenvalues is given by
		\begin{equation}\label{Eq:psuccess}
		    P_\text{success} \geq \left(\frac{1-\epsilon_R}{\kappa}\right)^2.
		\end{equation}

	\end{lemma}
	
	Lem.~\ref{condoverall} implies that for HHL to achieve an error $\epsilon_R$, assuming everything is exact except for the eigenvalue inversion and the representation of the eigenvalues, we can take 
	\begin{equation}\label{Eq:nlchoice}
	n_l=\mathcal{O}\left(\log\left(\frac{\kappa^2}{\epsilon_R} \right)\right),
	\end{equation}
	and 
	\begin{gather}
	    d = \left\lfloor\log\left(1+\frac{16.23\sqrt{\ln^2(r)+(\pi/2)^2}\kappa(2\kappa-\epsilon_R)}{\epsilon_R}\right)\right\rfloor \\
	    = \mathcal{O}\left(\log\left(n_l\frac{\kappa^2}{\epsilon_R} \right)\right)\nonumber
	    =\mathcal{O}\left(\log\left(\frac{\kappa^2}{\epsilon_R} \right)+\log\log\left(\frac{\kappa^2}{\epsilon_R} \right)\right).
	\end{gather}
	
	Thus, the overall complexity of the eigenvalue inversion to achieve an error of $\epsilon_R$ in the HHL algorithm is given by 
	\begin{equation}\label{Eq:rotation_complexity}
	    \mathcal{O}\left(\log^3(\kappa^2/\epsilon_R)\right).
	\end{equation}
	Leveraging amplitude amplification, Eq.~\ref{Eq:psuccess} also implies that the expected number of times we have to repeat the eigenvalue inversion is $\mathcal{O}\left(\kappa/(1-\epsilon_R)\right)$.
	
	As a final note, it might seem counter intuitive that by decreasing the error we also decrease the expected number of runs.
	This is due to a worst case analysis because the rotation step is an approximation, i.e. instead of $(\dots)\ket{0} + C\ket{x}\ket{1}$ we compute $(\dots)\ket{0} + \left(C\ket{x}+\epsilon\right)\ket{1}$. 
	In the very worst case, the errors are such that the absolute value of each component of $C\ket{x}$ is decreased. 
	This means that the probability of success, which corresponds to the $2$-norm of $C\ket{x} + \epsilon$, is in the worst case decreased by a quantity related to $\epsilon$. 
	However, this is compensated because the smaller $\epsilon_R$, the more costly a single run of the algorithm gets. 

	\section{Observables}\label{obsersec}
	
	In this section, we discuss how to estimate properties of the solution to the linear system based on the quantum state returned by the HHL algorithm.
	Depending on the properties of the linear system and the choice of observable it is possible to estimate the result efficiently or not.
	
	Throughout this section we need to distinguish between three scales:
	the normalized solution $\ket{x} = \sum_{i=0}^{N-1}x_{i}\ket{i}$, the solution to the original problem $\bm{x}$, and the solution to the problem with a normalized right-hand-side $\bm{b}/\norm{\bm{b}}$, given by $\bm{x'} = \bm{x}/\norm{\bm{b}}$.

	In the following, we show how to extract information about $\bm{x'}$ and $\bm{x}$ from the quantum state $\ket{x}$. 
	In particular, we are interested in
	\begin{itemize}
	    \item the solution norm $\norm{\bm{x}}$,
	    \item $F_{B}(\bm{x})\coloneqq\bm{x}^{T}B\bm{x}$, where $B\in\mathbb{R}^{N\times N}$ is a tridiagonal symmetric Toeplitz matrix, and
	    \item the absolute average of the components of the unscaled solution
	    $\abs{\frac{1}{N}\sum_{ i = 0 }^{N-1}\bm{x}_{i}},$ where $\bm{x}_i$ denotes the $i$-th element of $\bm{x}$.
	\end{itemize}
	
	The solution norm can be used to recover the real value of other observables, since we only have access to the normalized solution vector $\bm{x}/\norm{\bm{x}}$.
	The average is an example of a linear output functional and is related to the calculation of mean values, such as the mean temperature of a heat conduction model.
	Finally, $F_{B}(\bm{x})$ is an example of a quadratic output functional, also known as the energy norm.
	
	Evaluating $\bra{x}M\ket{x}$ for a linear operator $M$ as well as evaluating the solution norm have already been mentioned in the original HHL paper \cite{Harrow2009}, and calculating $\bm{u}^{T}\bm{x}$ for real vectors $\bm{u}$ is discussed in \cite{Zhikuan}. 
	However, these are generic results and do not discuss how to actually implement concrete cases, other than $M$ given by a weighted sum of Pauli terms.
	In contrast, we show how to construct the corresponding quantum circuits and analyze when the observables can be evaluated efficiently.
	
	Note that in HHL, since we need to load the normalized right-hand-side $\bm{b}/\norm{\bm{b}}$, we usually do not get access to properties of $\bm{x}$ but of $\bm{x'}$.
	Furthermore, as mentioned in Sec.~\ref{introsec} and pointed out in \cite{aaronsonfine}, state preparation can diminish the quantum advantage unless $\bm{b}$ follows the \emph{uniform assumption}, i.e., $\norm{\bm{b}}/\norm{\bm{b}}_\infty=\Omega(\sqrt{N})$. Then, state preparation can be achieved efficiently, e.g., with the methods presented in Sec.~\ref{statsect}.
	Suppose now that we can efficiently estimate $\norm{\bm{x'}}$ using samples from measuring the quantum state prepared by the HHL algorithm. 
    Due to the sampling error, we need $\mathcal{O}(1/\epsilon^{2})$ samples to estimate $\norm{\bm{x'}}$ with accuracy $\epsilon$. 
	Hence, to estimate $\norm{\bm{x}} = \norm{\bm{b}}\norm{\bm{x'}}$ with the same accuracy $\epsilon$, we need $\norm{\bm{b}}^{2}$-times the number of samples, which means that we need additional assumptions on the asymptotic behaviour of $\norm{\bm{b}}$.
	In the following, we discuss the introduced observables, specifying in each case whether the \emph{uniform assumption} suffices to guarantee an efficient algorithm, and whenever it does not, specifying what additional restrictions we need to impose on $\norm{\bm{b}}$.
	
    To improve readability we will let $\ket{\psi_0}$ refer to the quantum state we obtain if we do not check whether state preparation was successful, and $\ket{\psi_1}$ to the quantum state obtained in the case we run HHL once we know that we have prepared the right state.
    Furthermore, in the last measurement it is always implicit that we check that the eigenvalue register is in the $\ket{0}_{n_l}$ state. We refer to Sec.~\ref{condsect} and \ref{overanalysis} for a detailed analysis of the success probability of this step.
    More precisely, using $\bm{b}= \norm{\bm{b}} \sum_{j=0}^{N-1}\beta_j\ket{u_j}_{n_b}$, we consider the quantum states
    \begin{align}\label{psiobv}
	\ket{\psi_0}\coloneqq
	&\sum _ { j = 0 } ^ { N-1 } \left(\ldots \ket{0} + \frac{c\beta _ { j } \norm{\bm{b}}}{\sqrt{N}\norm{\bm{b}}_\infty}\ket{1}\right)\nonumber \\ 
	&\times\left( \sqrt { 1 - \frac { C^2 } { \tilde{\lambda}_ { j } ^ { 2 } } } \ket{0} + \frac { C } { \tilde{\lambda} _ { j } } \ket{1} \right) \ket{u _ { j } }_{n_{b}}, 
	\end{align}
	and
	\begin{equation}\label{psiobv1}
	\ket{\psi_1}\coloneqq
	\sum _ { j = 0 } ^ { N-1 } 
 	\beta _ { j }\left( \sqrt { 1 - \frac { C^2 } { \tilde{\lambda}_ { j } ^ { 2 } } } \ket{0} + \frac { C } { \tilde{\lambda} _ { j } } \ket{1} \right)  \ket{u _ { j } }_{n_{b}}, 
	\end{equation}
	where $\tilde{\lambda} _ { j }$ is an approximation to the $j^{\text{th}}$ eigenvalue, $\lambda _ { j }$, the first qubit is the ancilla for state preparation, and the second qubit is the ancilla for eigenvalue inversion.
	Depending on the chosen state preparation technique, the scaling factor $c > 0$ is set to $1$ or $<1$, cf. Sec.~\ref{statsect}.
	Note that in our case $C=\lambda_{\min}$, because in Lem.~\ref{condoverall} we set $x=N_l\lambda_j t/2\pi$. Within this section, we will treat it as a constant but will analyze how it affects the overall complexity in Sec.~\ref{overanalysis}.
	
    Then, as stated in the following proposition and proven in \cite{Harrow2009}, $\norm{\bm{x}}$ can be estimated using the probability of seeing $\ket{1}$ when measuring the ancilla qubit from the conditioned rotation of the inverse eigenvalues. Here we differentiate two cases: checking whether state preparation is successful before proceeding with the HHL algorithm, and at the end measuring the eigenvalue inversion ancilla qubit, or running always the full algorithm and at the end calculating the probability of seeing $\ket{11}$ in the state preparation and eigenvalue inversion ancilla qubits.
	All proofs that are not written explicitly in this section can be found in Appendix~\ref{Appendix:observables}.
	
	\begin{proposition}\label{normobv}
		Suppose we first check that state preparation has been successful before running HHL, and let $M_{1}\coloneqq \ket{1}\bra{1}\otimes I_{n_{b}}  $. Then
		\begin{equation}\label{Eq:norm_p1}
		P_{1} \coloneqq \bra{\psi_1}M_{1}^{\dagger}M_{1}\ket{\psi_1} = \frac{C^2\norm{\bm{x}}^{2}}{\norm{\bm{b}}^2}.
		\end{equation}
		On the other hand, suppose we run the full algorithm, i.e., now $M_{11}\coloneqq \ket{11}\bra{11}\otimes I_{n_{b}}  $. Then
		\begin{equation}\label{Eq:norm_p11}
		P_{11} \coloneqq \bra{\psi_0}M_{11}^{\dagger}M_{11}\ket{\psi_0} = \frac{c^{2}C^2\norm{\bm{x}}^{2}}{N\norm{\bm{b}}_\infty^2}.
		\end{equation}
	\end{proposition}

    As mentioned before, we estimate these quantities up to an accuracy $\epsilon$, which means that rescaling $P_1$ or $P_{11}$ to calculate $\norm{\bm{x}}$ also rescales $\epsilon$. 
    Therefore, to be able to estimate $\norm{\bm{x}}$ with an exponential speedup, we need to assume $\norm{\bm{b}}=\mathcal{O}(1)$ for Eq.~\ref{Eq:norm_p1} and $\norm{\bm{b}}_\infty=\mathcal{O}\left(\frac{1}{\sqrt{N}}\right)$ for Eq.~\ref{Eq:norm_p11}. It can be easily seen that both conditions are equivalent under the \emph{uniform assumption}.
    To achieve a polynomial speedup, these assumptions could be relaxed, as discussed in Sec.~\ref{overanalysis}.

    Now we focus on $F_B(\bm{x})$. Let $p,q\in\mathbb{R}$ be the diagonal and off-diagonals entries of $B$, respectively. Then, 
	\begin{equation}
	F_{B}(\bm{x})\coloneqq\bm{x}^TB\bm{x}=p\sum_{ i = 0 }^{N-1}x_{i}^{2}+2q\sum_{ i=0 }^{N-2}x_{i}x_{i+1},
	\end{equation}
	where $\bm{x}\in\mathbb{ R }^{N}$.
	We now proceed to show how to calculate $F_{B}( \bm{x'} )$ using $n_{b}$ different observables. 
	Each observable is constructed by appending a different set of gates at the end of the HHL circuit and by measuring a different set of qubits. 
	The quantities obtained from each measurement are stored and combined classically to estimate $F_{B}( \bm{x'} )$.
	
	The first observable consists of applying a single Hadamard gate on the last qubit, $\ket{q_{0}}_{n_{b}}$, and measuring the qubit in the computational basis. We denote by $n_{1}(0)$ the probability of observing $\ket{0}$ and by $n_{1}(1)$ the probability of observing  $\ket{1}$.
	
	The $k^{\text{th}}$ observable, $k>1$, is constructed by appending the circuit from Fig.~\ref{hadik} at the end of the algorithm and measuring the last $k$ qubits of the solution register in the computational basis. 
	Similarly as before, let $\ket{\phi_{1,k}}$ the state resulting from applying the observable circuit to $\ket{\psi_1}$ ($\ket{\phi_{0,k}}$ when applied to $\ket{\psi_0}$), and let $n_{k}(0) $ and $n_{k}(1)$ denote the probability observing the states $\ket{0}\ket{1}_{k-1}$ and $\ket{1}\ket{1}_{k-1}$ (last $k-1$ qubits all 1), respectively, when measuring $\ket{q_{k-1}}, \ldots, \ket{q_{0}}$.
	
	\begin{figure}[h!]
		\includegraphics[width=1\linewidth]{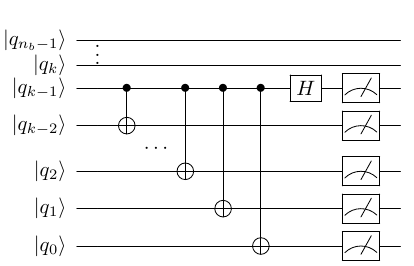}
		\caption{The $k^{\text{th}}$ observable for calculating $F_{B}(\tilde{ x })$.}
		\label{hadik}
	\end{figure}
	
	\begin{proposition}\label{eigerr}
		Suppose we check state preparation has been successful before running HHL, and let $M_{k}(i)\coloneqq \ket{1}\bra{1}\otimes I_{n_{b}-k} \otimes \ket{i}\bra{i}\otimes \ket{1}_{ k-1}\bra{1}_{ k-1}$, for $i=0,1$. Then $n_{k}(0)-n_{k}(1)$ equals
		\begin{align}
		&\bra{\phi_{1,k}}M_{k}(0)^{\dagger}M_{k}(0)\ket{\phi_{1,k}} 
		-\bra{\phi_{1,k}}M_{k}(1)^{\dagger}M_{k}(1)\ket{\phi_{1,k}}\nonumber \\
		&=\frac{2C^2}{\norm{\bm{b}}^2}\sum_{ \substack{i = 0 \\i \equiv 2^{k-1}-1\ (\textrm{\normalfont mod}\ 2^{k}) }}^{N-1}x_{i}x_{i+1}.
		\end{align}	
		On the other hand, suppose we run always the full algorithm, i.e., now $M_{k}(i)\coloneqq \ket{11}\bra{11}\otimes I_{n_{b}-k} \otimes \ket{i}\bra{i}\otimes \ket{1}_{ k-1}\bra{1}_{ k-1}$, for $i=0,1$. Then $n_{k}(0)-n_{k}(1)$ equals
		\begin{align}
		&\bra{\phi_{1,k}}M_{k}(0)^{\dagger}M_{k}(0)\ket{\phi_{1,k}}
		-\bra{\phi_{1,k}}M_{k}(1)^{\dagger}M_{k}(1)\ket{\phi_{1,k}}\nonumber \\
		&=\frac{2c^2C^2}{N\norm{\bm{b}}_\infty^2}\sum_{ \substack{i = 0 \\i \equiv -1\ (\textrm{\normalfont mod}\ 2^{k}) }}^{N-2}x_{i}x_{i+1}.
		\end{align} 
	\end{proposition}
	Since $p$ and $q$ are known parameters, and $\sum_{ i = 0 }^{N-1}\abs{x_{i}}^{2}/\norm{\bm{b}}^2=\norm{ \bm{x'} }^{2}$ can be calculated from Prop.~\ref{normobv}, we have shown a method to compute $F_{B}(\bm{x'} )$ requiring $n_{b}=\log(N)$ different observables.
	As before, if we are interested in $F_{B}(\bm{x} )$, we need to add the extra assumption $\norm{\bm{b}}=\mathcal{O}(1)$, or $\norm{\bm{b}}_\infty=\mathcal{O}\left(\frac{1}{\sqrt{N}}\right)$ to the \emph{uniform assumption}.
	
	Finally, the solution average can be calculated by appending the circuit shown in Fig.~\ref{hadmeas} and measuring each qubit from the solution register in the computational basis in addition to measuring the ancilla qubit.
	
	\begin{figure}[h!]
		\includegraphics[width=0.8\linewidth]{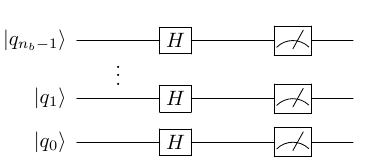}
		\caption{Hadamards to compute the average of the solution.}
		\label{hadmeas}
	\end{figure}
	
	\begin{proposition}\label{Prop:obv2}
		Suppose we check state preparation has been successful before running HHL, and let  $M_{1,0}\coloneqq \ket{1}\bra{1}\otimes\ket{0}_{n_{b}}\bra{0}_{n_{b}}  $ and $\ket{\phi_1}\coloneqq\left(I\otimes H^{\otimes n_{b}}\right)\ket{\psi_1}$. 
		Then
		\begin{equation}
		\bra{\phi_1}M_{1,0}^{\dagger}M_{1,0}\ket{\phi_1}=\abs{\frac{C}{\sqrt{N}\norm{\bm{b}}}\sum_{ i = 0 }^{N-1}x_{i}}^{2}.
		\end{equation}
		On the other hand, suppose we run always the full algorithm, i.e., now $M_{1,0}\coloneqq \ket{11}\bra{11}\otimes\ket{0}_{n_{b}}\bra{0}_{n_{b}}$ and $\ket{\phi_0}\coloneqq\left(I_2\otimes H^{\otimes n_{b}}\right)\ket{\psi_0}$.
		Then
		\begin{equation}
		\bra{\phi_0}M_{1,0}^{\dagger}M_{1,0}\ket{\phi_0}=\abs{\frac{cC}{N\norm{\bm{b}}_\infty}\sum_{ i = 0 }^{N-1}x_{i}}^{2}.
		\end{equation}
	\end{proposition}
	
	Therefore, to be able to estimate the average efficiently we need to assume $\norm{\bm{b}}_\infty^2=\mathcal{O}\left(1\right)$ in addition to the \emph{uniform assumption}.
	As pointed out before, under the \emph{uniform assumption}, this is equivalent to assume $\norm{\bm{b}}=\mathcal{O}(\sqrt{N})$.
	Note that this is a significantly weaker assumption than required for $\norm{\bm{x}}$ or $F_B(\bm{x})$.
	
	Finally, as shown in Appendix~\ref{Appendix:extrap_properties}, the Richardson extrapolation technique can also be applied directly to the observables.
	That is, we can calculate the value of an observable of interest, $O$, from the intermediate values for such observable, $O_j$ for $j=1,\dots,l$, obtained from HHL approximating $e^{iAtk}$ with $V(t,m_j(k))^k$ as
	\begin{equation}
	    O = \sum_{j=1}^l a_j O_j.
	\end{equation}
	In the following, we analyze the (expected) use of resources of the complete algorithm for each of the observables presented in this section and conclude under which assumptions a quantum advantage is possible.

	\section{Overall analysis}\label{overanalysis}
	
	In this section we provide an overall error and complexity analysis of the algorithm.
	We derive optimal parameters for $t$ and $n_l$ as well as the minimum required accuracy for each block of the algorithm to achieve an overall target accuracy.
	The section concludes with an investigation of the cases where a quantum advantage can be achieved depending on the norm of the right hand side and the considered observable.
	
	First, we derive an expression for $t$. 
	For QPE to be accurate, we need to set $t$ such that $\abs{\frac{\lambda_j t}{2\pi}}\in[0,1)$.
	The eigenvalues of tridiagonal Toeplitz matrices are given by \cite{tridieigs}
	\begin{equation}
	\lambda _ { j } = a - 2 b \cos \left( \frac { j \pi } { N + 1 } \right),
	\end{equation}
	where $N$ denotes the dimension.
	Therefore, we can compute $\lambda_{\max}$, which is bounded independently of $N$ by $|\lambda_{\max}| \leq a + 2|b|$, and set
	\begin{equation}\label{time}
	t = 2\pi\cdot\frac {   2 ^ { n_{l} } - 1 } { 2 ^ { n_{l} }\lambda _ { \max } }<\frac{2\pi }{\lambda _ { \max }}.
	\end{equation}

	The remaining parameters to determine are the accuracy of each block of the algorithm, i.e., $\epsilon_S$ for state preparation, $\epsilon_A$ for Hamiltonian simulation, and $\epsilon_R$ for eigenvalue inversion. To do so, we will need the following theorem, which results from a straight application of triangle inequalities and the results obtained in Lemmas~\ref{psperr}, \ref{hamoverall}, and \ref{condoverall}.
	\begin{theorem}\label{hhlerr}
		Let $\ket{x}$ denote the exact solution, and $\ket{\tilde{x}}$ the solution returned by the HHL algorithm using the approximations described in Sections~\ref{statsect}-\ref{condsect}. Then
		\begin{equation}
		\norm{\ket{x}-\ket{\tilde{x}}} \leq \epsilon_S+ 2\epsilon_{ A } + \epsilon_{ R } + \mathcal{ O }(\epsilon_A ^2)
		\end{equation}
		for $\epsilon_{ A } \rightarrow 0$.
	\end{theorem}
	
	Therefore, to achieve an overall error $\epsilon$, it is enough to choose, e.g.,
	\begin{equation}
	    \epsilon_s=\frac{\epsilon}{3},\quad \epsilon_A = \frac{\epsilon}{6},\quad\text{and}\quad \epsilon_R = \frac{\epsilon}{3}.
	\end{equation}

	We now proceed to analyze the total complexity of the algorithm.
	The complexity of state preparation and the conditional rotation are given in Eq.~\ref{Eq:state_prep_complexity} and Eq.~\ref{Eq:rotation_complexity}, respectively.
	The complexity arising from QPE and Hamiltonian simulation with Richardson extrapolation, as given in Eq.~\ref{Eq:qpe_parallel_complexity} needs to be reformulated by substituting the value obtained for $n_l$ (Eq.~\ref{Eq:nlchoice}), yielding
	\begin{equation}
	\mathcal{ O }\left(\log(\kappa^2/\epsilon)\log(N)/\sqrt{\epsilon}\right),
	\end{equation}
	for running the algorithms in parallel.
	A summary of the complexity analysis can be found in Table~\ref{totalgates}, were the expected number of gates uses Lem.~\ref{condoverall}, i.e., that the probability of success is at least $\left(\frac{1-\epsilon}{\kappa}\right)^2$, and amplitude amplification.

	\begin{table*}[ht]
		\caption{Total gate count for the HHL algorithm for a target accuracy $\epsilon$ and different observables. 
		The 'expected runs' row under complexity reflects the probabilistic nature of state preparation and eigenvalue inversion, and leverages amplitude amplification.
		The 'scaling' column contains the scaling, $S$, of an obtained observable $F(\bm{x^\prime})$ with respect to the exact value $F(\bm{x})$, i.e., $S F(\bm{x^\prime}) = F(\bm{x})$.
		Thus, $S$ denotes the factor that amplifies the sampling error and implies the restrictions on $\bm{b}$ discussed in the main text.
		The $\lambda_{\min}$ factor comes from setting $C=\lambda_{\min}$, where $C$ is the constant defined for the conditional rotation step (Eq.~\ref{psiobv1}).
		For the considered tridiagonal Toeplitz matrices, we have $\lambda_{\min} \sim  1/\kappa$.
		The last row contains the number of circuits required to be run to calculate any of the observables, where the $1/\sqrt{\epsilon}$ factor comes from the extrapolation scheme and $1/\epsilon^2$ from the sampling complexity.
		All values correspond to a sequential computation of the different full HHL algorithms, i.e. without checking whether state preparation was successful before proceeding further so that amplitude amplification can be applied.
		Other settings can be derived similarly using the results given in this paper.}
		\resizebox{\textwidth}{!}{%
		\begin{tabular} {cccc} 
			\toprule
			Observable & \multicolumn{2}{c}{Complexity} & Scaling $S$\\
			\midrule
			
			\multirow{2}{*}{Norm} & $1$ circuit & $\mathcal{O}\left(\log^{d}(N) + \log(\kappa^2/\epsilon)\log(N)\sqrt{\epsilon} + \log^3(\frac{\kappa}{\epsilon})\right)$ & \multirow{2}{*}{$\frac{\norm{\bm{b}}}{\lambda_{\min}}$}\\
			\cmidrule{2-3}
			& expected runs & $1$ & \\
			
			\midrule
			
			\multirow{2}{*}{$F_B(\bm{x})$} & $1$ circuit & $\mathcal{O} \left(\log^{d}(N) + \log(\kappa^2/\epsilon)\log(N)\sqrt{\epsilon} + \log^3(\frac{\kappa}{\epsilon})+\log(N)\right)$ & \multirow{2}{*}{$\frac{\norm{\bm{b}}^2}{\lambda_{\min}^2}$} \\
			\cmidrule{2-3}
			& expected runs & $\mathcal{O}\left(\frac{\kappa}{1-\epsilon}\cdot\frac{\sqrt{N\kappa}\norm{\bm{b}}_\infty }{\sqrt{\epsilon}\norm{\bm{b}}}\right)$ & \\
			
			\midrule
			
			\multirow{2}{*}{Average} & $1$ circuit & same as norm & \multirow{2}{*}{$\frac{\norm{\bm{b}}}{\lambda_{\min}\sqrt{N}}$} \\
			\cmidrule{2-3}
			& expected runs & $\mathcal{O}\left(\frac{\kappa}{1-\epsilon}\cdot\frac{\sqrt{N\kappa}\norm{\bm{b}}_\infty }{\sqrt{\epsilon}\norm{\bm{b}}}\right)$ & \\
			
			\bottomrule
			
			\multirow{2}{*}{All} & number of circuits & \multirow{2}{*}{$\mathcal{O}\left(\frac{1}{\sqrt{\epsilon}\epsilon^2}\right)$} & \\
			& (extrapolation and sampling) & & \\
			
			\bottomrule
		\end{tabular}}
		\label{totalgates}
	\end{table*}
	
	To conclude this section, we analyze in which cases it is possible to obtain a quantum advantage. 
	The limiting factor is the norm of the right-hand-side, $\ket{b}$, as was already pointed out in \cite{aaronsonfine}. More precisely, we consider the following points:
	\begin{enumerate}[(i)]
		\item If state preparation is probabilistic, then the probability of success is $\Omega\left(\norm{\bm{b}}/(\sqrt{N}\norm{\bm{b}}_\infty)\right)$. Therefore, the expected number of runs to successfully prepare the state is $\mathcal{O}\left(\sqrt{N}\norm{\bm{b}}_\infty/\norm{\bm{b}}\right)$, c.f.~Lem.~\ref{psperr}.
        \label{item:quantum_advantage_state_prep}
		\item Since $\bm{b}$ is encoded in a quantum state, HHL solves the problem $A\bm{x^{\prime}}=\bm{b}/\norm{\bm{b}}$ and estimates $F(x')$, for some linear observable $F$, with accuracy $\epsilon$. 
		The solution can then usually be recovered by $F(\bm{x}) = \norm{\bm{b}}\left(F(\bm{x}^{\prime}) + \epsilon\right)$.
		\label{item:quantum_advantage_sampling}
	\end{enumerate}
	Under the \emph{uniform assumption}, (\ref{item:quantum_advantage_state_prep}) does not pose a problem.
	Regarding (\ref{item:quantum_advantage_sampling}), even though we have reduced the circuit complexity to $1/\sqrt{\epsilon}$, the results are calculated from probabilities of measuring one or more ancilla qubits in particular states, which are estimated by sampling.
	The number of samples (runs of HHL) needed to achieve an estimation error $\epsilon$ scales as $\mathcal{O}(1/\epsilon^2)$.
	If the result is rescaled by $\norm{\bm{b}}$, so is the error $\epsilon$.
	Hence, whenever we are interested in $\norm{\bm{x}}$ or $F_B(\bm{x})$, i.e., when we need to multiply the result by $\norm{\bm{b}}$, the exponential speedup only holds if $\norm{\bm{b}}=\mathcal{O}(1)$, or equivalently, under the \emph{uniform assumption}, $\norm{\bm{b}}_{\infty}=\mathcal{O}(1/\sqrt{N})$.
	However, when we are interested in the absolute average, these assumptions can be relaxed.
	In this case it is still possible to obtain an exponential speedup as long as $\norm{\bm{b}}=\mathcal{O}(\sqrt{N})$, or equivalently, under the \emph{uniform assumption}, $\norm{\bm{b}}_{\infty}=\mathcal{O}(1)$.
	
	Note that we were focussing on achieving an exponential speedup.
	The assumptions mentioned here can be relaxed further if we are targeting a polynomial speedup.
	Furthermore, note that only the sampling, i.e., the repeated evaluation of the quantum circuit, poses a potential problem, while the complexity of the circuits scale as $polylog(1/\epsilon)$.
	This implies that the algorithm can be parallelized extensively, which may be another possibility to achieve an advantage.
	
	\section{Results}\label{ressect}
	
	In this section we first show the results for a $8\times 8$ system of linear equations on a quantum simulator. Afterwards we show the values obtained on a real quantum device for the solution norm of a $4 \times 4$ system and the solution average of a $2\times 2$ system.
	
	\subsection{Quantum Simulator}
	For the quantum simulator, we consider the problem with parameters $n_{b}=3$, $a=2$, $b=-1/2$ and $\epsilon = 2^{-5} \approx 0.03$, which gives $c=0.1$ for the state preparation algorithm. We choose the initial state to be specified by the polynomial $p(x)=x^{3}-x^{2}+x+1$.
	The simulations, implemented with Qiskit \cite{Qiskit}, require $14$ qubits and evaluate the exact quantum states.
	
	First we run the state preparation algorithm alone and afterwards the full HHL algorithm.
	In the first case, the vector directly obtained from Qiskit's state vector simulation is $c\ket{b}/\sqrt{N}$, where the $1/\sqrt{N}$ factor comes from the initial Hadamard gates used to obtain a uniform superposition of all basis states. Since these constants are known, $\ket{b}$ can be recovered, as it is shown in Fig.~\ref{statefig}(a), where the vector obtained from the simulation is plotted against $p(i/7),0\leq i\leq 7$. From the Taylor expansion of $\sin(cp(i/7))/c$, we can calculate the error in approximating $p(i/7)$, plotted in Fig.~\ref{statefig}(b) against the error obtained from the simulation.
	It can be nicely seen that the target accuracy $\epsilon$ is achieved.

	\begin{figure}
		\centering
		\includegraphics[width=1\linewidth,height=\linewidth]{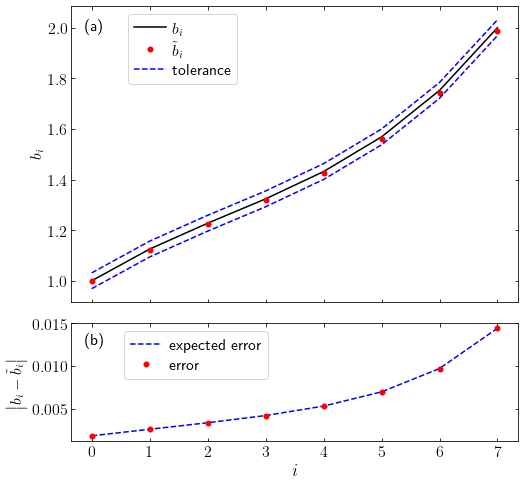}
		
		\caption{State preparation simulation. The $x$-axis denotes the basis states. 
		(a) The plot shows the vector obtained from the state vector simulation with Qiskit, and $p(x)$ refers to $p(i/7)$.
		(b) The expected error was calculated from the Taylor expansion of $\sine(cp(x))/c$. }
		\label{statefig}
	\end{figure}
	
	Similarly, the results from the state vector simulator for the complete HHL algorithm implementation are rescaled by $c/\sqrt{N}$. The solution vectors with their respective norms obtained from first running the full algorithm, and then running the HHL with Richardson extrapolation for the same tolerance, are shown in Fig.~\ref{solutionfig}(c)-(d). Although the results obtained for smaller step sizes are outside the tolerated error bounds, depicted as dashed lines in the plots, the extrapolated solution lies within these lines. The reason the vector obtained with the full algorithm is closer to the analytic solution is because the theoretical error bounds calculated for Richardson extrapolation are tighter, and we try to use the least resources that allow to achieve a given accuracy.
	
	\begin{figure}
		\centering
		
		\includegraphics[width=1\linewidth,height=0.99\linewidth]{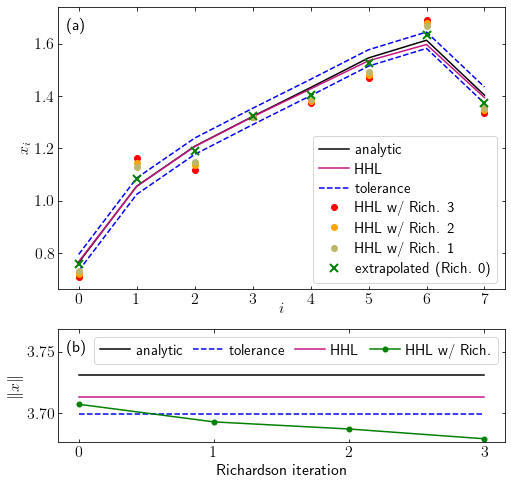}
		
		\caption{Result of complete HHL algorithm.
		(a) Full state vector obtained from simulator. The $x$-axis denotes the basis states. 'HHL' denotes the solution obtained running the full algorithm with Trotter exponent $m(1)=5$ and 'Rich. $3$', 'Rich. $2$', 'Rich. $1$' correspond to the solutions obtained with $m(1)=2,3,4$, respectively, while 'Rich. $0$' denotes the extrapolated solution. 
		(b) Norm of the solution vectors. The $x$-axis corresponds to the 'HHL w/ Rich.' vectors, and $0$ is the extrapolated solution.}
		\label{solutionfig}
		
	\end{figure}

	\subsection{Quantum Hardware}
	
	For the quantum hardware, we consider two problems that we run both on $4$ qubits: 
	a $4\times 4$ system, i.e., $n_b=2$, where we evaluate the solution norm, and a $2\times 2$ system, i.e., $n_b=1$, where we evaluate the absolute solution average. 
	We run tests for $10$ different initial states, each prepared by $R_{y}$ rotations with an angle $\theta$ applied to the initial state qubit(s).
	In both cases we use the parameters $n_{l}=2$, $a=1$, $b=-1/3$, and number of shots $M = 8192$.
	For the smaller system the fourth qubit is used for the eigenvalue inversion, while for the larger system we directly measure the eigenvalues, and perform their inversion classically.
	More precisely, for the latter, we measure $\ket{\lambda_j}_{n_l}$ in the computational basis, which gives us estimates of $\lambda_j$ as well as the corresponding occurrence probabilities that are equal to $|\beta_j|^2$.
	This allows us to classically compute 
	\begin{equation}
	    \sum_{j=0}^3 \frac{|\beta_j|^2}{\lambda_j^2},
	\end{equation}
	which equals $\norm{\bm{x}}$, since we have $\norm{\bm{b}} = 1$ by construction.
	Note that this approach is only suitable for a small demonstration, but will not scale to larger problem instances as it would require an exponentially increasing number of shots.
	
	Note that in both cases we do not need to uncompute the QPE.
	For the larger system we actually just run a QPE without any further steps and estimate the resulting eigenvalues.
	For the smaller system, we can find $t$ for the Hamiltonian evolution such that the QPE will be exact as it only has two eigenvalues.
	Here, the omitted inverse QPE results in a solution that is scaled by $1/\sqrt{2^{n_l}}$ due to the Hadamards.

	All the experiments are evaluated on the \emph{ibmq$\_\!$boeblingen} $20$-qubit backend provided by IBM Quantum. 
	Both circuits require $4$ fully connected qubits, however, the connectivity of the device does not allow to choose such arrangement.
	Therefore, in both cases, one SWAP gate is required to implement the circuits.
	The connectivity of the quantum device as well as the corresponding circuits are provided in Fig.~\ref{connectivity}-\ref{realcirc2} in Appendix~\ref{realhardappendix}.
	
	The experiments are run on a noisy device with gate and readout errors.
	Qiskit \cite{Qiskit} allows to mitigate the readout errors by individually preparing and measuring all basis states, a detailed treatment on the topic can be found in \cite{dewes, bravyi2020mitigating}. 
	To deal with the gate errors, we use another Richardson extrapolation to extract the result at the zero noise limit as it was done in \cite{cloudqcomp,OptionPricing, vazquez2020efficient}. 
	We focus the error mitigation on the CNOT gates since they have an average randomized benchmarking fidelity of $97.8$\% compared to $99.7$\% for the single-qubit gates.
	More precisely, we run each experiment three times: First the original circuit, and then substituting each CNOT gate by $3$ and $5$ CNOT gates, respectively.
	Mathematically, the three circuits have the same effect, however, on real hardware this systematically amplifies the CNOT gate error, allowing a zero noise extrapolation.
	The simulated results can be nicely recovered using the real quantum hardware.
	
    Fig.~\ref{realplot}(a) shows the results obtained for calculating the norm of the solution vector, i.e., the probability of measuring $\ket{1}$ in the conditional rotation qubit. Fig.~\ref{realplot}(b) shows the absolute solution average obtained by applying a Hadamard gate to the solution qubit, and estimating the probability of $\ket{0}$ in the conditioned rotation qubit. 
	
	\begin{figure}
		\centering
		\includegraphics[width=1\linewidth]{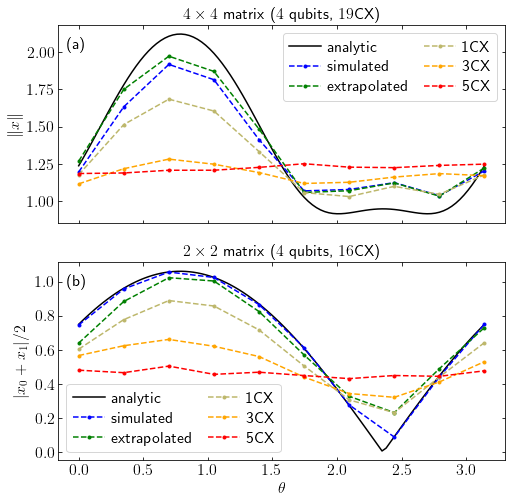}
		\caption{Real hardware on $4$ qubits. Tests for $10$ different initial states, each prepared by a $R_{y}$ rotation by $\theta$ of the initial state qubit(s).
		(a) Norm of the solution vector $\bm{x}=(x_{0},x_{1})^{T}$. 
		(b) Absolute average of the solution vector $\abs{x_{0}+x_{1}}/2$.}
		\label{realplot}
	\end{figure}
	
	\section{Conclusion}
	
	In this paper we provided a detailed implementation and analysis of the HHL algorithm for linear systems of equations defined by tridiagonal Toeplitz matrices and right-hand-sides given by analytic functions.
	For every step of the algorithm we introduce novel techniques and prove corresponding error and complexity bounds, and we combine these results to draw conclusions about the overall algorithm.
	Our main result allows for a quadratic reduction of the circuit complexity required by Hamiltonian simulation within QPE, and thus, constitutes an important step to resolve the main bottleneck of the algorithm.
	In addition, we study different observables and analyze necessary conditions to estimate them efficiently, i.e., achieving an exponential quantum advantage.
	Particularly for the absolute solution average, we find that weaker conditions are sufficient, increasing the applicability of the algorithm.
	Although quantum hardware is not ready yet to run the HHL algorithm for problems of practically relevant size, we demonstrated the algorithm for small systems using simulation as well as real quantum devices.
	
	The introduced extrapolation scheme to reduce the complexity of the Hamiltonian simulation circuits is analyzed for Toeplitz matrices.
	However, it is straight-forward to apply it in more general settings, e.g., with more then two non-commuting terms in Hamiltonian decomposition or for other algorithms leveraging Hamiltonian simulation.
	Theoretically the scope of this scheme could also be widened to reduce general algorithmic errors, provided these can be written as a power series depending on an adjustable step size.
	However, in the cases of more than two non-commuting terms or other methods than product formulas the resulting error bounds remain open questions that we refer to future research.
	
	\section{Acknowledgements}
	
	We thank Albert Frisch, Dominik Steenken, and Harry Barowski for the fruitful discussions on the eigenvalue inversion used within the HHL algorithm.
	
	We would also like to acknowledge the support of the National Centre of Competence in Research \textit{Quantum Science and Technology} (QSIT).
	
	IBM, the IBM logo, and ibm.com are trademarks of International Business Machines Corp., registered in many jurisdictions worldwide. Other product and service names might be trademarks of IBM or other companies. The current list of IBM trademarks is available at \url{https://www.ibm.com/legal/copytrade}.

	\bibliographystyle{IEEEtran}
	\bibliography{MyCollection}
	
	\appendix
	\section{\textit{d}-dimensional Poisson equation}\label{poissonapp}

	Cao et al. give in \cite{Cao} a quantum algorithm to solve the Poisson equation in $d$ dimensions. In their work they show that the implementation reduces to running $d$ parallel circuits solving the one-dimensional Poisson equation. Running our algorithm with $a=2$, $b=-1$ and $\log(N)$ qubits gives the solution to the one dimensional case with $h=1/(N+1)$ discretization step size. Below we reproduce their findings to illustrate how to modify our algorithm to solve the Poisson equation in $d$ dimensions.
	
	Discretizating the Laplacian 
	\begin{equation}
	\Delta = \sum_{k=1}^{d}\frac{\partial^{2}}{\partial x_{k}^{2}}
	\end{equation}
	on a grid with mesh size $h=1/(N+1)$ using divided differences leads to a system of linear equations
	\begin{equation}
	-\Delta_{h}\vec{v}=\vec{f}_{h}.
	\end{equation}
	$\Delta_{h}$ is a symmetric positive definite matrix and can be expressed as a tensor product with $d$ terms
	\begin{align}
	h^{2}\Delta_{h}&= A_{h}\otimes I_{n_{b}}\otimes \cdots \otimes I_{n_{b}}\\
	&+I_{n_{b}}\otimes A_{h}\otimes \cdots \otimes I_{n_{b}}\\
	&+\cdots + I_{n_{b}}\otimes\cdots\otimes I_{n_{b}}\otimes A_{h},
	\end{align}
	where $n_{b}=\log(N)$ and $A_{h}$ is the $N\times N$ tridiagonal symmetric matrix with $a=2$ and $b=-1$. Then
	\begin{equation}\label{dpoisson}
	e^{\mathrm{i}\Delta_{h}t}=e^{\mathrm{i}h^{-2}A_{h}t}\otimes\cdots\otimes e^{\mathrm{i}h^{-2}A_{h}t}.
	\end{equation}
	is a tensor product with $d$ terms. We have shown how to simulate $e^{\mathrm{i}h^{-2}A_{h}t}$, therefore, Eq.~\ref{dpoisson} shows that simulating $e^{\mathrm{i}\Delta_{h}t}$ can be achieved by using $d$ registers with $n_{b}$ qubits each and running in parallel the circuit for $e^{\mathrm{i}h^{-2}A_{h}t}$ in each register.  
	
	\section{Technical proofs from Sec.~\ref{statsect}}\label{Appendix:StatePreparation}
	\subsection{Proof of Lemma \ref{lemma:poly_circuit}}\label{Sec:lemma:poly_circuit}
	\begin{proof}
	Since $d\leq \lceil \frac{n}{2}\rceil$, Lem.~8 from \cite{Itendefs} applies, and we can implement a $k-$controlled $R_{y}$ gate with at most $(16k-12)$ CNOTs for all $k\in\{1,\dots,d\}$.
		
		The circuit for $P_p$ consists of all possible $1-$, $2-$,...,$d-$ controlled $R_{y}$ gates. Therefore we may count the total number of CNOTs as 
		\begin{equation}
		\sum_{k=1}^{d}\binom{n}{k}(16k-12)<\sum_{k=1}^{d}16\frac{n^{k}}{(k-1)!}
		<(16e) n^{d},
		\end{equation}
		where the last inequality was derived from the Taylor series of $e^{x}$.
	\end{proof}
	
	\subsection{Proof of Lemma \ref{statexp}}\label{Sec:statexp}
	\begin{proof}
        Let $x_i = \frac{i}{N-1}$ for $i=0,\dots,N-1$. Then from the Taylor expansion of $\sin^2(x)$ and using that $p_f(x_i)\in[-1,1]$,
		\begin{align}
		&\bra{\psi}F_{c}^{\dagger}M_{1}^{\dagger}M_{1}F_{c}\ket{\psi}
		=
		\frac{1}{N} \sum_{i=0}^{N-1} \sin^2(cp_f(x_i))\\
		&=\frac{1}{N} \sum_{i=0}^{N-1}c^2p^2_f(x_i)+\mathcal{O}\left(c^4p^4_f(x_i)\right)\\
		&\geq\frac{1}{N}\left(\sum_{i=0}^{N-1}c^2\frac{f^2(x_i)}{\norm{\bm{b}}^2_\infty}\right)+c^2\epsilon_p+\mathcal{O}\left(c^4\right)\\
		&= \frac{c^2\norm{\bm{b}}^2}{N\norm{\bm{b}}^2_\infty}-c^2\epsilon_p+\mathcal{O}\left(c^4\right).
		\end{align}

	\end{proof}
	
	\subsection{Proof of Lemma \ref{psperr}}\label{Sec:psperr}
	\begin{proof}
	    First, from the Taylor series of $\sin(x)$ and using that $\norm{\bm{x}}_2\leq\sqrt{N}\norm{\bm{x}}_\infty$ for $\bm{x}\in\mathbb{R}^{N}$, we have that
	    \begin{equation}
	        \ket{\tilde{b}}\coloneqq\frac{\tilde{\bm{b}}}{\norm{\tilde{\bm{b}}}}=\frac{1}{\sqrt{N}} \sum_{i=0}^{N-1} \frac{\sin(cp_f(x_i))}{\norm{\tilde{\bm{b}}}}
	        =\frac{\frac{c\bm{b}}{\sqrt{N}\norm{\bm{b}}_\infty}+\bm{\epsilon}}{\norm{\tilde{\bm{b}}}},
	    \end{equation}
	    where $\norm{\bm{\epsilon}}\leq c\epsilon_1 +c^3 + \mathcal{O}\left(c^5\right)$.
	    Therefore, 
	    \begin{align}
	        &\norm{\frac{\bm{b}}{\norm{\bm{b}}}-\frac{\tilde{\bm{b}}}{\norm{\tilde{\bm{b}}}}}
	        =\norm{\frac{\frac{c\bm{b}}{\sqrt{N}\norm{\bm{b}}_\infty}}{\frac{c\norm{\bm{b}}}{\sqrt{N}\norm{\bm{b}}_\infty}}-\frac{\tilde{\bm{b}}}{\norm{\tilde{\bm{b}}}}}\\
	        &\leq\frac{2\norm{\bm{\epsilon}}}{\frac{c\norm{\bm{b}}}{\sqrt{N}\norm{\bm{b}}_\infty}}\leq\frac{2\sqrt{N}\norm{\bm{b}}_\infty\left(\epsilon_p+c^2\right)}{\norm{\bm{b}}}+\mathcal{O}\left(\frac{\sqrt{N}\norm{\bm{b}}_\infty c^4}{\norm{\bm{b}}}\right).
	    \end{align}
		Hence, if $\bm{b}$ satisfies the \emph{uniform assumption}, the above can be simplified to $2\left(\epsilon_1+c^2\right)+\mathcal{O}\left(c^4\right)$. 
		Then, writing $\ket{b}\coloneqq\bm{b}/\norm{\bm{b}}$, the overall error can be calculated similarly as in Sec.~III.B of \cite{Montanaro2016QuantumAA} by
		\begin{align}
		&\norm{\ket{x}-\ket{\tilde{x}}}\leq 
		\norm{\frac{A^{-1}\ket{b}}{\norm{A^{-1}\ket{b}}}-\frac{A^{-1}\ket{\tilde{b}}}{\norm{A^{-1}\ket{\tilde{b}}}}}\\
		 &\leq \norm{\frac{A^{-1}\ket{b}\bigg(\norm{A^{-1}\ket{\tilde{b}}}-\norm{A^{-1}\ket{b}}\bigg)}{\norm{A^{-1}\ket{b}}\norm{A^{-1}\ket{\tilde{b}}}}-\frac{A^{-1}\bm{\epsilon}}{\norm{A^{-1}\ket{\tilde{b}}}}} \\
		& \leq 2\frac{\norm{A^{-1}\bm{\epsilon}}}{\norm{A^{-1}\ket{\tilde{b}}}}
		\leq 2\kappa \left[2\left(\epsilon_p+c^2\right)+\mathcal{O}\left(c^4\right)\right].
		\end{align}
		The first inequality comes from the triangle inequality applied twice and the last inequality from the definition of the condition number written as 
		\begin{equation*}
		\kappa = \max\frac{\frac{\norm{A^{-1}\bm{\epsilon}}}{\norm{A^{-1}\ket{\tilde{b}}}}}{\frac{\norm{\bm{\epsilon}}}{\norm{\ket{\tilde{b}}}}}.
		\end{equation*}
		
	\end{proof}
	
	\section{Hamiltonian simulation circuits}\label{hamcircapp}
	Let $I _ { k }$ denote the $2 ^ { k } \times 2 ^ { k }$ dimensional identity matrix. Following the nomenclature from Qiskit, we use the $X=\left( \begin{array} { c c } { 0 } & { 1 } \\ { 1 } & { 0 } \end{array} \right)$, $U_{1}(\lambda)$ and $R _ { x } ( \theta )$ quantum gates, where
	\begin{equation*}
	U _ { 1 } ( \lambda ) = \left( \begin{array} { c c } { 1 } & { 0 } \\ { 0 } & { e ^ { \mathrm{i} \lambda } } \end{array} \right)  \text { and }  R _ { x } ( \theta ) = \left( \begin{array} { c c } { \cos \frac { \theta } { 2 } } & { - \mathrm{i} \sin \frac { \theta } { 2 } } \\ { - \mathrm{i} \sin \frac { \theta } { 2 } } & { \cos \frac { \theta } { 2 } } \end{array} \right).
	\end{equation*}
	Computing $e ^ { \mathrm{i} H _ { i } t }$ in each case yields the following circuits.
	\begin{enumerate}[(1)]
		\item \begin{equation}
		e ^ { \mathrm{i} H _ { 1 } t } = e ^ { \mathrm{i} a t }I _ { n_{b} } 
		\end{equation}
		
		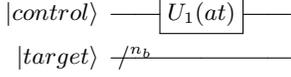
\begin{figure}[h!]
			\centering
			\[
			\Qcircuit @C=1em @R=1em {	
				\lstick{\ket{control}}        & \qw & \gate{U_{1}(at)}    & \qw  & \qw     \\
				\lstick{\ket{target}} & {/^{n_{b}}} \qw & \qw & \qw & \qw
			}
			\]
			\caption{Circuit for implementing a controlled $e^{\mathrm{i}H_{1}t}$ on $n_{b}$ qubits.}
			\label{}
		\end{figure}
		
		\item \begin{equation}
		e ^ { \mathrm{i} H _ { 2 } t } =  I _ { n_{b}-1 } \otimes R _ { x } ( - 2 b t )
		\end{equation}

		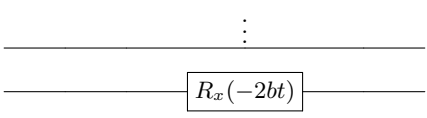
\begin{figure}[h!]
			\centering
			\[
			\Qcircuit @C=2.5em @R=1em {	
				& \qw & \qw & \qw                & \qw & \qw\\
				&     &     & \vdots             &     &\\
				& \qw & \qw & \qw                & \qw & \qw\\ 
				& \qw & \qw & \gate{R_{x}(-2bt)} & \qw & \qw 
			}
			\]
			\caption{Circuit for implementing $e^{\mathrm{i}H_{2}t}$.}
			\label{a2circ}
		\end{figure}
        
		\item
		The general case for $H_{3}$ will make use of $\mathcal{ O }(n_{b})$  CNOTs, the full circuit is shown in Fig.~\ref{h3nbfig}. The key steps of the circuit are:
		\begin{enumerate}[1.]
			\item Flag the states $\ket{0}_{n_b}$ and $\ket{N-1}_{n_b}$ setting an ancilla qubit to $\ket{0}$.
			
			On condition of the flag being $\ket{1}$:
			\item $\ket{i}_{n_b}\mapsto \ket{i+1}_{n_b}$.
			\item Apply $R _ { x } ( -2bt)$ to $q_{0}$.
			\item $\ket{i-1}_{n_b}\mapsto \ket{i}_{n_b}$.
		\end{enumerate}
		\begin{figure}[h!]
			\centering
			\[
			\Qcircuit @C=1em @R=0.8em {	
				& \ctrl{1} & \ctrlo{1} & \multigate{4}{+1} & \qw          & \multigate{4}{-1} & \ctrl{1} & \ctrlo{1} & \qw \\
				& \ctrl{2} & \ctrlo{2} & \ghost{+1}        & \qw          & \ghost{-1}        & \ctrl{2} & \ctrlo{2} & \qw \\
				& & &  &  & &\\
				& \ctrl{1} & \ctrlo{1} & \ghost{+1}        & \qw          & \ghost{-1}        & \ctrl{1} & \ctrlo{1} & \qw \\
				& \ctrl{2} & \ctrlo{2} & \ghost{+1}        & \gate{R_{x}( -2bt)} & \ghost{-1}        & \ctrl{2} & \ctrlo{2} & \qw \\
				\\
				\lstick{\ket{1}}
				&\targ     & \targ     & \ctrl{-2} & \ctrl{-2} & \ctrl{-2} &\targ & \targ & \qw
			}
			\]
			\caption{Schematic representation of the circuit implementing $e^{\mathrm{i}H_{3}t}$ for the general dimension case.}
			\label{h3nbfig}
		\end{figure}
		Nonetheless, for a small enough $n_b$ there is a more efficient implementation requiring $\mathcal{ O }(n_{b}^{2})$ CNOTs. Both implementations yield the same results, however the method described below is faster for all the values of $n_{b}$ we could simulate in a local memory. The full circuit is given in Fig.~\ref{cblockscirc}.
		\begin{figure}[h!]
			\centering
			\[
			\Qcircuit @C=2.5em @R=1em {	
				& \qw & \qw & \qw & \multigate{4}{C_{n_{b}-1}} & \qw\\
				& & & \cdots & &\\
				& \qw & \multigate{2}{C_{2}} & \qw & \ghost{C_{n_{b}-1}} & \qw\\ 
				& \multigate{1}{C_{1}} & \ghost{C_{2}} & \qw & \ghost{C_{n_{b}-1}} & \qw\\
				& \ghost{C_{1}} & \ghost{C_{2}} & \qw & \ghost{C_{n_{b}-1}} & \qw
			}
			\]
			\caption{Schematic representation of the circuit implementing $e^{\mathrm{i}H_{3}t}$ for the general dimension case.}
			\label{cblockscirc}
		\end{figure}
		The circuit for each $C_{j}$-block, $1 \leq j \leq n_{b} - 1$, is shown in Fig.~\ref{singlecblock}.
		
		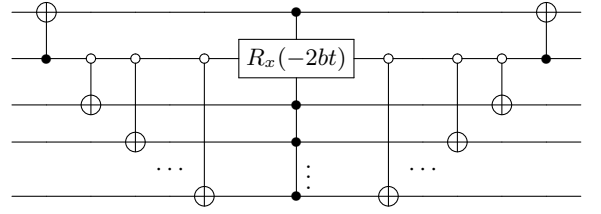
\begin{figure}[h!]
			\centering
			\[
			\Qcircuit @C=1em @R=0.7em {
				&\targ &\qw&\qw&\qw&\qw&\ctrl{1}&\qw&\qw&\qw&\qw&\targ &\qw\\
				&\ctrl{-1}&\ctrlo{1}&\ctrlo{2}&\qw&\ctrlo{4}&\gate{R_{x}(-2b t)}&\ctrlo{4}&\qw&\ctrlo{2}&\ctrlo{1}&\ctrl{-1}&\qw\\
				&\qw&\targ&\qw&\qw&\qw&\ctrl{-1}&\qw&\qw&\qw&\targ&\qw&\qw\\
				&\qw&\qw&\targ&\qw&\qw&\ctrl{-1}&\qw&\qw&\targ&\qw&\qw&\qw\\
				& & &  & \cdots& &\quad\vdots& & \cdots& & & &  \\
				&\qw&\qw&\qw&\qw&\targ&\ctrl{-2}&\targ&\qw&\qw&\qw&\qw&\qw\\
			}
			\]
			\caption{Detailed circuit for $C_{j}$.}
			\label{singlecblock}
		\end{figure}
	\end{enumerate}
	
	\section{Technical proofs from Sec.~\ref{hamsect}}\label{Appendix:Hamiltonian}
	
	\subsection{Proof of Lemma \ref{hamerreq}}\label{Sec:hamerreq}
	\begin{proof}
		Let $A=H_1 + H_2 + H_3$ and $H=H_2 + H_3$. Then, with the notation from Sec.~\ref{hamsect}, we want to bound the quantity
		\begin{equation}
		\norm{e^{\mathrm{i}Atk}-V(t,m)^k}.
		\end{equation}
		From the Taylor series expansion one finds that
		\begin{equation}
		S_1(t/m) = e^{\mathrm{i}Ht/m} + E(t/m),
		\end{equation}
		where, for $t/m\rightarrow 0$,
		\begin{equation}
		E ( t / m ) : = \frac { 1 } { 6 } \left[ \left[ H _ { 2 } , H _ { 3 } \right] , \frac { 1 } { 4 } H _ { 2 } + \frac { 1 } { 2 } H _ { 3 } \right] \left( \frac { \mathrm{i} t } { m } \right) ^ { 3 } + \mathcal { O } \left( \frac { t^4 } { m^4 } \right) .
		\end{equation}
		Here $\left[ H _ { 2 } , H _ { 3 } \right] = H _ { 2 } H _ { 3 } - H _ { 3 } H _ { 2 }$ denotes the matrix commutator. 
		Both $e^{\mathrm{i}Ht}$ and $S_1(t)$ are unitary operators, hence of norm $1$. 
		Since $S_1(t/m)^{mk} = S_1(tk/mk)^{mk}$, and using the Cauchy-Schwarz and triangle inequalities gives
		\begin{align}
		&\norm{e^{\mathrm{i}Atk}-e ^ {\mathrm{i}H_{1}tk} S_1(t/m)^{mk}}\leq
		\norm{e^{\mathrm{i}Htk} - S_1(tk/mk)^{mk}}\\ 
		&\leq mk\norm{e^{\mathrm{i}Ht/m} - S_1(t/m)}
		=mk\norm{E (t/ m )}\\
		&\leq k\frac {t ^ { 3 }b^{3} } { 2m ^ { 2 } } + \mathcal { O } \left( \frac { t^4 } { m^4 } \right) \quad \text{for }t/m\rightarrow 0.
		\end{align}
		
	\end{proof}
	
	\subsection{Proof of Lemma \ref{hamoverall}}\label{Sec:hamoverall}
	\begin{proof}
		In this proof we assume exact procedures for state preparation and conditioned rotation of the eigenvalues. 
		We will denote $U_{1}$ the unitary matrix corresponding to the application of the powers of $e^{\mathrm{i}At}$; $U_{2}$ the unitary matrix corresponding to the $\mathrm { QFT } ^ { \dagger }$, inversion of eigenvalues and QFT; and $U_{3}$ to the application of powers of the inverse of $V$. 
		Similarly, $\tilde { U } _ { 1 }$ and $\tilde { U } _ { 3 }$ will denote the matrices corresponding to the same parts of the algorithm but using $V(t,m)$ ($U_{2}$ is the same in both cases).
		
		Then we can write
		\begin{equation}
		\ket{x} = U _ { 3 } U _ { 2 } U _ { 1 } \ket{b},
		\end{equation}
		and
		\begin{equation}\label{Eq:hhl_decomposition}
		\ket{\tilde{x}} = \tilde { U } _ { 3 } U _ { 2 } \tilde { U } _ { 1 } \ket{b}.
		\end{equation}
		We can express $V(t,m(k))^k$ as
		\begin{equation}
		V(t,m(k))^k=e^{\mathrm{i}Atk}+E,\quad\text { where } \quad \norm{E} < \epsilon _ { A }.
		\end{equation}
		Then, for $N = 2 ^ { n _ { l } }, n_{l}\in\mathbb{ N }$,
		\begin{align}
		\tilde { U } _ { 1 } &= 
		\frac { 1 } { \sqrt { N } } \left( \sum _ { k = 0 } ^ { N- 1 } \ket{k}\bra{k} \otimes \left( e^{\mathrm{i}Atk} + E \right) \right) \\
		&= U _ { 1 } + \frac { 1 } { \sqrt { N } } \left( \sum _ { k = 0 } ^ { N- 1 } \ket{k}\bra{k}\right) \otimes E,
		\end{align}
		where
		\begin{equation}
		    \norm{\frac { 1 } { \sqrt { N } } \left( \sum _ { k = 0 } ^ { N- 1 } \ket{k}\bra{k}\right) \otimes E} = \norm{E}.
		\end{equation}
		Similarly,
		\begin{equation}
		\tilde { U } _ { 3 } = U _ { 3 } + \frac { 1 } { \sqrt { N } } \left( \sum _ { k = 0 } ^ { N- 1 } \ket{k}\bra{k} \right)\otimes E.
		\end{equation}
		Thus, expanding and using that the $U_{i}$ are unitary,
		\begin{align}
		\norm{U _ { 3 } U _ { 2 } U _ { 1 } - \tilde { U } _ { 3 } U _ { 2 } \tilde { U } _ { 1 }}&\leq
		2\norm{E}+\norm{E}^2\\
		&< 2\epsilon_A +\mathcal{O}\left(\epsilon_A^2\right)\quad\text{for }\epsilon _ { A }\rightarrow 0.\label{177}
		\end{align}
	\end{proof}
	
	\section{Gate count from estimating the eigenvalues}\label{richapp}
	Here we provide a detailed analysis of the total number of CNOTs arising from the part of the HHL algorithm that estimates the eigenvalues (i.e.~QPE), first in the original case and then we compare the complexity with our extrapolation scheme. 
	
	Let $G$ denote the number of CNOTs required in one application of the circuit for Hamiltonian simulation, and $\epsilon$ denote the overall error from the HHL algorithm. Recall that QPE involves simulating $e^{\mathrm{i}Atk}$ for $k=2^0,2^1,\dots,2^{n_l-1}$, for which we set (from Sec.~\ref{hamsect})
	\begin{equation}
	m(2^k) = \mathcal{O}\left(\sqrt{\frac{2^{k}}{2\epsilon}}\right).
	\end{equation}
	Thus, the complexity of QPE is
	\begin{equation}
	    \sum_{k=0}^{n_l-1}Gm\left(2^k\right)=\mathcal{O}\left(\frac{G}{\sqrt{2\epsilon}}\sum_{k=0}^{n_l-1}2^{\frac{k}{2}}\right)
	    =\mathcal{O}\left(\frac{G}{\sqrt{2\epsilon}}2^{\frac{n_l}{2}}\right).
	\end{equation}
	On the other hand, as in Sec.~\ref{richsect}, let $\vec{m}_l = (1,\dots,l)$, $l=1/\sqrt{\epsilon}$, and $m_j(k)=\left(\lfloor k^{\frac{3}{n_l}}+1\rfloor\right)j$.
	Then, the total number of CNOTs arising from the $l$ QPE modules that we run for the extrapolation scheme can be calculated as
	\begin{align}
	\sum_{j=1}^{l}\sum_{k=0}^{n_l-1}Gm_j\left(2^k\right) &= \sum_{j=1}^{l}G\sum_{k=0}^{n_l-1}\left(\lfloor 2^{\frac{3k}{n_l}}\rfloor +1 \right)j\\
	&<\sum_{j=1}^{l}G\sum_{k=0}^{n_l-1}8j
	= 4Gn_l l(l+1)
	\end{align}
	Since $n_l = \mathcal{O}\left(\log(\kappa^2 /\epsilon)\right)$ (Sec.\ref{overanalysis}), and $G=\mathcal{ O }(n_{b})=\mathcal{ O }(\log(N))$, the overall complexity of estimating the eigenvalues is in the original case
	\begin{equation}
	    \mathcal{O}\left(\log(N)\kappa^2/\epsilon\right),
	\end{equation}
	while with our extrapolation scheme allows to run the circuits in parallel, reducing this complexity to $\mathcal{O}(\log(N)\kappa^2/\sqrt{\epsilon})$.

	\section{Technical proofs from Sec.~\ref{richsect}}\label{Appendix:richext}
	\subsection{Proof of Theorem \ref{Thm:extraperr}}\label{Sec:Thm:extraperr}
	\begin{proof}
    Using the notation introduced in Sec.~\ref{richsect}, here we prove the bounds for $\norm{e^{\mathrm{i}Ht}-V_{l}(t,\vec{m}_l)}$. First we introduce some definitions and results from \cite{Chin2010}.
	
	The choice of a symmetric product splitting, such as the symmetric Strang splitting $S_1(t)=e^{\frac{1}{2}\mathrm{i}H_2 t}e^{\mathrm{i}H_3 t}e^{\frac{1}{2}\mathrm{i}H_2 t}$, is important because then the error term has only odd powers of $t$. That is,
	\begin{equation}
	    S_1(t) = \exp\left(\mathrm{i}t(H_2 + H_3) + (\mathrm{i}t)^3 E_3 + (\mathrm{i}t)^5 E_5 + \dots\right),
	\end{equation}
	where the $E_i$ are higher order error commutators of $H_2$ and $H_3$. Then the $2l$-order approximation of $e^{\mathrm{i}t(H_2 + H_3)}$ can be bounded as follows.

	\begin{proposition}[\cite{Chin2010}]\label{Prop:Chin}
	Let $l\in\mathbb{N}$ denote the number of extrapolation points, $V_{l}(t, \vec{m}_l)$ as defined in Eq.~\ref{Eq:Vldef} and \ref{ajdef}, and $C\coloneqq t\norm{H_2+H_3}$. Then
	\begin{equation}\label{Eq:Chin}
	    \norm{e^{\mathrm{i}Ht}-V_{l}(t,\vec{m}_l)}\leq \abs{e_{2l+1}}t^{2l+1}\norm{E_{2l+1}},
	\end{equation}
	where 
	\begin{equation}
	    e_{2l+1} = (-1)^{l-1}\prod_{i=1}^{l}\frac{1}{m_i^2}.
	\end{equation}
	\end{proposition}
	
	From the Taylor expansions of $e^{\mathrm{i}Ht}$ and $V_{l}(t,\vec{m}_l)$ one can see that $\norm{E_k}\leq \frac{(2d)^{k}}{k!}$, where $d=\max\{\norm{H_2},\norm{H_3}\}$. However note that this is a loose upper bound obtained by successive applications of the triangle inequality.
	
	\begin{proposition}\label{Prop:lambert}
	Let $\vec{m}_l=(1,\dots,l)$, and let $d$ and $H$ be defined as in Prop.~\ref{mchoice}. Then $\norm{e^{\mathrm{i}Ht}-V_{l}(t,\vec{m}_l)}<\epsilon$ for
	\begin{equation}\label{Eq:lambert}
	    l=\frac{\ln(2dt/\epsilon)}{4W\left(\frac{\ln(2dt/\epsilon)}{2e\sqrt{4dt}}\right)},
	\end{equation}
	where $W(x)$ denotes the Lambert function (cf.~\cite{Corless1996}) and $\ln(x)$ the natural logarithm.
	\end{proposition}
	
	\begin{proof}
	From Prop.~\ref{Prop:Chin}, we have that for $\vec{m}_l=(1,\dots,l)$,
	\begin{equation}
	    \norm{e^{\mathrm{i}Ht}-V_{l}(t,\vec{m}_l)}\leq\frac{(2dt)^{2l+1}}{(2l+1)!l!^2}.
	\end{equation}
	Since $\frac{1}{n!}<\frac{1}{\sqrt{2\pi n}}\left(\frac{e}{n}\right)^n$,
	\begin{align}
	    \frac{(2dt)^{2l+1}}{(2l+1)!l!^2}&<\frac{(2dt)^{2l+1}}{\sqrt{2\pi(2l+1)} 2\pi l}\left(\frac{e}{2l+1}\right)^{2l+1}\left(\frac{e}{l}\right)^{2l}\\
	    &<\frac{t}{2\sqrt{2\pi}}\left(\frac{4e\sqrt{dt}}{4l}\right)^{4l},
	\end{align}
	where for the last inequality we assume $d\leq l^{5/2}$.
	
	Finally, some calculation shows that the solution to the equation
	\begin{equation}
	    \frac{t}{2\sqrt{2\pi}}\left(\frac{4e\sqrt{dt}}{4l}\right)^{4l} = \epsilon
	\end{equation}
	is 
	\begin{equation}
	    l=\frac{\ln(t/2\sqrt{2\pi}\epsilon)}{4W\left(\frac{\ln(t/2\sqrt{2\pi}\epsilon)}{4e\sqrt{dt}}\right)}.
	\end{equation}
	\end{proof}
	Finally, if we use that $dt$ is bounded, take the worst case $tk = t 2^{n_l}$, use that $n_l = \log(1/\epsilon)$ and $m_j(k)=\left(\lfloor k^{\frac{3}{n_l}}\rfloor +1\right) m_j$, we get
	\begin{align}
	    &\norm{e^{\mathrm{i}Ht}-V_{l}(tk,\vec{m}_l(k))}\leq\frac{(2dtk)^{2l+1}}{(2l+1)!l!^2 8^{2l}}\\
	    &<\left(\frac{e^2 2dtk}{8(2l+1)l}\right)^{2l}\frac{2tdke}{(2l+1)2\pi l \sqrt{2\pi(2l+1)}}\\
	    &<\left(\frac{e\sqrt{2dtk}}{4l}\right)^{4l}\frac{1}{\epsilon}
	\end{align}
	Similarly as before this gives
	\begin{equation}
	    l = \frac{\ln(1/\epsilon)}{2W\left(\frac{2\ln(1/\epsilon)}{e\sqrt{2dt/\epsilon}}\right)},
	\end{equation}
	for which $l=1/\sqrt{\epsilon}$ is a good approximation using that $W(x)\approx x$ as $x\rightarrow 0$.
	\end{proof}
	
	\subsection{Proof of Theorem \ref{thm:hhl_extrapolation}}\label{Sec:thm:hhl_extrapolation}
	\begin{proof}
	Suppose that
	\begin{equation}
	    \norm{e^{\mathrm{i}Ht}-V_l(t,\vec{m}_l)}<\epsilon_A.
	\end{equation}
	Then, with the notation from the proof of Lem.~\ref{hamoverall},
	\begin{equation}\label{Eq:Vkeps_A}
	    \norm{e^{\mathrm{i}Atk}-\sum_{j=1}^{l} a_j V(t,m_j(k))^k}<\epsilon_A,
	\end{equation}
	where we recall that $A = H_1 + H$, $H_1$ commutes with $H$, and $V(t,m_j(k))^k = e^{\mathrm{i}H_1 tk}S_1\left(t/m_j(k)\right)^{km_j(k)}=e^{\mathrm{i}Atk} + E_j$. Then, from Eq.~\ref{Eq:hhl_decomposition} we have
	\begin{gather}
	    \sum_{j=1}^{l} a_j \text{HHL}\left(V\left(t,m_j(k)\right)^{k},\ket{b}\right)
	    =\sum_{j=1}^{l} a_j \tilde { U }^{(j)} _ { 3 } U _ { 2 } \tilde { U }^{(j)} _ { 1 } \ket{b}\\
	    =U_3 U_2 U_1\ket{b} + \sum_{j=1}^{l} a_j \left(E_j U_2 U_1 + 2E_j U_2 E_j + U_3 U_2 E_j\right),
	\end{gather}
	where $\tilde { U }^{(j)}_i = U_i + E_j$, and the $j$ subscript specifies that the approximation depends on $m_j(k)$.
	
	Using that the $U_i$ are unitary and from Eq.~\ref{Eq:Vkeps_A} gives
	\begin{equation}
	    \norm{\sum_{j=1}^{l} a_j \left(E_j U_2 U_1 + 2E_j U_2 E_j + U_3 U_2 E_j\right)}<2\epsilon_A+\mathcal{O}(\epsilon_A^2).
	\end{equation}
	
	Finally, since $\sum_{j=1}^l a_j=1$, we can use the same argument as at the end of the proof of Lem.~\ref{hamoverall} and obtain the result.
	\end{proof}

	\section{Technical proofs from Sec.~\ref{condsect}}\label{Appendix:conditionalrot}
	\subsection{Proof of Lemma \ref{lemma:chebychev}}\label{Sec:propcheb}
		\begin{proof}
		The error of Chebychev interpolation in approximating a function $f$ on $[-1,1]$ by $p_f$, in the case that $f$ possesses an analytic extension to a complex neighbourhood $\mathcal{D}$ of $[-1,1]$, can be bounded by \cite{TAD86}
		\begin{equation}\label{Eq:cheb_analytical}
		    \norm{f-p_f}_{L^{\infty}([-1,1])}
		    \leq \frac{2\abs{\gamma}}{\pi}\frac{1}{\left(\rho^{d+1}-1\right)\left(\rho +\rho^{-1}-2\right)}\cdot\max_{z\in\gamma}\abs{f(z)},
		\end{equation}
		where $d$ is the degree of $p$ and for $\rho>1$, $\gamma\subset\mathcal{D}$ is the ellipsis defined by
		\begin{equation}
		    \gamma\coloneqq\{z=\cos(\theta-\mathrm{i}\log\rho),0\leq\theta\leq2\pi\}.
		\end{equation}
		From Eq.~\ref{Eq:cheb_analytical}, the error on $[a,b]$ can be bounded by
		\begin{align}\label{Eq:cheb_analyticalab}
		    &\norm{f-p_f}_{L^{\infty}([a,b])}=\norm{\hat{f}-p_{\hat{f}}}_{L^{\infty}([-1,1])}\\
		    &\leq \frac{2\abs{\gamma}}{\pi}\frac{1}{\left(\rho^{d+1}-1\right)\left(\rho +\rho^{-1}-2\right)}\cdot\max_{z\in\gamma}\abs{\hat{f}(z)},
		\end{align}
		where $\hat{f}\coloneqq f\circ\Phi$ and $\Phi:[-1,1]\mapsto [a,b]$ is the linear map defined by $\Phi(x)= a +\frac{1}{2}(x+1)(b-a)$.
		Since $\hat{f}(x)$ has a singularity at $\Phi(x)=0$, and the leftmost point of $\gamma$ is at $-(\rho + \rho^{-1})/2$, one can see that in Eq.~\ref{Eq:cheb_analyticalab} we require
		\begin{equation}
		    \rho \leq \frac{4a}{b-a}+1.
		\end{equation}
		Note that a simple calculation now shows that choosing to interpolate over the interval $[a,N_l-1]$ would require to take $d=\Omega\left(N_l\log(1/\epsilon)\right)$, thus we will focus on a piecewise interpolation.
		We will divide the interval $[a,N_l-1]$ into $M$ subintervals $[a_i,a_{i+1}]$, with $a_1 = a$ and $a_{M+1}=N_l-1$, and perform a piecewise polynomial interpolation on each.
		Since the goal is a constant degree $d$ over the whole domain, we will fix $\frac{4a_i}{a_{i+1}-a_i}=1$ and take $\rho = 2$.
		Therefore, $a_{i+1}=5a_i$ and $M=\lceil\log_5(\frac{N_l-1}{a})\rceil$.
		
		We can now proceed to give a more specific bound for Eq.~\ref{Eq:cheb_analytical}. First, for an interval $[a,b]$, $\max_{z\in\gamma}\abs{f(z)}$ is attained at $z=-(\rho + \rho^{-1})/2 = -5/4$, and
		\begin{align}
		    \abs{f\left(\Phi(-5/4)\right)}&=\abs{\arcsin(\frac{C}{a +\frac{1}{2}\left(-\frac{5}{4}+1\right)(b-a)})}\\
		    &=\abs{\arcsin(\frac{2C}{a})}.
		\end{align}
		\begin{proposition}
		Let $x\in\mathbb{R}$, then
		\begin{equation}
		    \abs{\arcsin(x)}=\sqrt{\ln^2(r)+\left(\frac{\pi}{2}\right)^2},
		\end{equation}
		where $r=x+\sqrt{\abs{1-x^2}}$.
		\end{proposition}
		\begin{proof}
		For $x\in\mathbb{R}$ we can write
		\begin{equation}
		    \arcsin(x) = \frac{1}{\mathrm{i}}\ln(\mathrm{i}x+\sqrt{\abs{1-x^2}}e^{\frac{\mathrm{i}}{2}\text{Arg}(1-x^2)}).
		\end{equation}
		Then, for $\abs{x}>1$,
		\begin{equation}
		    \text{Arg}\left(1-x^2\right)=\pi\implies e^{\frac{\mathrm{i}}{2}\text{Arg}(1-x^2)}=\mathrm{i}.
		\end{equation}
		Hence,
		\begin{align}
		    \arcsin(x) &= \frac{1}{\mathrm{i}}\ln(\mathrm{i}\left(x+\sqrt{\abs{1-x^2}}\right))\\
		    &= \frac{1}{\mathrm{i}}\ln(\left(x+\sqrt{\abs{1-x^2}}\right)e^{\mathrm{i}\pi/2}).
		\end{align}
		\end{proof}
		Thus, writing $r\coloneqq \frac{2C}{a}+\sqrt{\abs{1+\left(\frac{2C}{a}\right)}}$, we have $\max_{z\in\gamma}\abs{f(z)}\leq\sqrt{\ln^2(r)+\left(\frac{\pi}{2}\right)^2}$ and
		\begin{align}
		    \norm{f-p_f}_{L^{\infty}([a,b])}&\leq\frac{2\abs{\gamma}}{\pi}\frac{2\sqrt{\ln^2(r)+\left(\frac{\pi}{2}\right)^2}}{2^{d+1}-1}\\
		    &\leq\frac{8.13\sqrt{\ln^2(r)+\left(\frac{\pi}{2}\right)^2}}{2^{d+1}-1}
		\end{align}
		means that for an accuracy $\epsilon$ we can fix the degree of the polynomial $d=\log(1+\frac{8.13\sqrt{\ln^2(r)+\left(\frac{\pi}{2}\right)^2}}{\epsilon})-1$ on each subinterval.
		\end{proof}
	\subsection{Proof of Lemma \ref{condoverall}}\label{Sec:condoverall}
	\begin{proof}
	     Let $p(x)= \frac{C}{x}+\epsilon_C$ for $x\in[a,N_l-1]$ be the polynomial approximation we can achieve, where $C$ will be chosen so that $C/x\in(0,1]$.
    	 For readability here we will assume we have applied QPE$^{\dagger}$ and that we have measured $\ket{1}$ in the ancilla qubit for the conditioned rotation, and $\ket{0}_{n_l}$ in the eigenvalue register, denoting the probability of this event as $P_\text{success}$.
    	 We have to include the later since in reality the eigenvalue register will not be perfectly uncomputed, i.e. the amplitude of each $\ket{l}_{n_l}$ will be small but nonzero. 
	    Under this assumptions, we have the quantum state 
	    \begin{equation}
	        \frac{\sum_{j=0}^{N-1}\beta_j\left(\sum_{l=0}^{N_l-1}\abs{\alpha_{l|j}}^2p(l)\right)\ket{0}_{n_l}\ket{u_j}_{n_b}\ket{1}}{P_\text{success}},
	    \end{equation}
	    where \cite[\S 5.2.1]{Nielchu}
	    \begin{equation}
	        \alpha_{l|j} \coloneqq \frac{1}{N_l}\left(e^{2\pi \mathrm{i}\left(\frac{\lambda_j t}{2\pi}-\frac{l+\tilde{\lambda}_j}{N_l}\right)}\right)^k,
	    \end{equation}
	    and $\tilde{\lambda}_j$ is a $n_l$-bit binary approximation to $\lambda_j t / 2\pi$.
	    
	    From \cite[\S 5.2.1]{Nielchu}, and writing $\Bar{\lambda}_j\coloneqq\frac{N_l\lambda_j t}{2\pi}$ we have that
	    \begin{equation}
	        \sum_{l=0}^{N_l-1}\abs{\alpha_{l|j}}^2=1\quad\text{ and } \sum_{\abs{l-\Bar{\lambda}_j}>e}\abs{\alpha_{l|j}}^2\leq\frac{1}{2(e-1)}\eqqcolon \epsilon_1,
	    \end{equation}
	    where $e\in\mathbb{Z}$ characterizes the desired error, so we will choose it as $e=2^{n_1}$, where $n_l = n_1 + n_2 + n_3$ will be specified later.
	    
	    Mapping the eigenvalues with $t$ so that $\Bar{\lambda}_j\coloneqq N_lt\lambda_j/2\pi\in[2^{n_1+n_2},N_l-1]$
	    gives for $\abs{l-\Bar{\lambda}_j}\coloneqq\abs{\delta_j}\leq e$
	    \begin{equation}
	        \abs{\frac{C^{\prime}}{l}-\frac{C^\prime}{\Bar{\lambda}_j}}= \frac{C^\prime\delta_j}{\Bar{\lambda}_j(\Bar{\lambda}_j+\delta_j)}\leq \frac{C}{2^{n_2}-1}\eqqcolon\epsilon_2,
	    \end{equation}
	    where $C^\prime$ is some constant we will specify later and $C\coloneqq \frac{C^\prime 2\pi}{N_l t\lambda_{\min}} $.
	    Finally, since we know the eigenvalues have been mapped to $[2^{n_1 + n_2},N_l-1]$, we can write
    	\[
        p(l)= 
            \begin{cases}
                \frac{C^\prime}{l} +\epsilon_C,& \text{if } 2^{n_1 + n_2}\leq l\leq N_l -1\\
                1,              & \text{if } 1\leq l < 2^{n_1 + n_2}.
            \end{cases}
        \]
	    Then, from the triangle inequality, and using that $1/\Bar{\lambda}_j\leq 1$ and that $C^\prime / l\in(0,1]$, we obtain the bound
	    \begin{align}
	        \abs{\frac{C^\prime}{\Bar{\lambda}_j}-\sum_{l=0}^{N_l-1}\abs{\alpha_{l|j}}^2p(l)}
	        &\leq \abs{\frac{C^\prime}{\Bar{\lambda}_j}-\sum_{\abs{l-\Bar{\lambda}_j}\leq e}\abs{\alpha_{l|j}}^2\frac{C^\prime}{l}}+\epsilon_1+\epsilon_C\\
	        &\leq \epsilon_2 +\epsilon_1+\epsilon_1\epsilon_2+\epsilon_1 C+\epsilon_C \eqqcolon\epsilon_R.
	    \end{align}
	    
    	Assuming everything else is exact within HHL, and relabeling $C\lambda_{\min}$ as $C$, we obtain the following approximation to the solution vector
		\begin{equation}
		    \sum_{j=0}^{N-1}\beta_j\left(\frac{ C}{\lambda_j }+\epsilon_R\right)\ket{u_j}\eqqcolon CA^{-1}_{\epsilon_R}\ket{b}.
		\end{equation}
		
		Now we will analyze the overall accuracy of the algorithm assuming everything is exact except for the binary representation of the eigenvalues and the conditioned rotation, i.e. the algorithm returns the quantum state $\bm{\tilde{x}}/\norm{\bm{\tilde{x}}}\coloneqq A^{-1}_{\epsilon_R}/\norm{A^{-1}_{\epsilon_R}}$.
		
		As in \cite{demmelmatrix}, substracting the perturbed equation from $A\bm{x}=\ket{b}$ and rearranging gives
		\begin{equation}
		    A^{-1}_\epsilon\ket{b}-A^{-1}\ket{b}=A^{-1}\left(A-A_\epsilon\right)A^{-1}_\epsilon\ket{b}.
		\end{equation}
		Then, applying the triangle inequality yields
		\begin{equation}
		    \frac{\norm{A^{-1}\ket{b}-A^{-1}_\epsilon\ket{b}}}{\norm{A^{-1}_\epsilon\ket{b}}}\leq
		    \frac{\norm{A^{-1}-A^{-1}_\epsilon}}{\norm{A^{-1}_\epsilon\ket{b}}}\leq\norm{A^{-1}}\norm{A-A_\epsilon},
		\end{equation}
		where $\norm{A^{-1}}=1/\lambda_{\min}$, and
		\begin{equation}
		    \norm{A-A_\epsilon}\leq  \abs{\frac{\epsilon_R\lambda_{\max} ^2}{C(1-\epsilon_R\lambda_{\min} /C)}}.
		\end{equation}
		Using the triangle inequality one more time and denoting $E=A^{-1}-A^{-1}_\epsilon$ we get
		\begin{align}
		&\norm{\ket{x}-\ket{\tilde{x}}}= \norm{\frac{A^{-1}\ket{b}}{\norm{A^{-1}\ket{b}}}-\frac{A^{-1}_{\epsilon}\ket{b}}{\norm{A^{-1}_{\epsilon}\ket{b}}}}  \\
		&=\norm{\frac{A^{-1}\ket{b}\left(\norm{A^{-1}_{\epsilon}\ket{b}}-\norm{A^{-1}\ket{b}}\right)-E\ket{b}\norm{A^{-1}\ket{b}}}{\norm{A^{-1}\ket{b}}\norm{A^{-1}_{\epsilon}\ket{b}}}}\\
		&\leq 2\frac{\norm{E}}{\norm{A^{-1}_{\epsilon}\ket{b}}} \leq 2\norm{A^{-1}}\norm{A-A_\epsilon}\\
		&\leq\abs{2\frac{\epsilon_R\lambda_{\max} ^2}{\lambda_{\min}C(1-\epsilon_R\lambda_{\min} /C)}}.
		\end{align}
		Finally, using that $\norm{A_\epsilon}\leq\frac{\lambda_{\max}}{1-\epsilon_R \lambda_{\min}/C}$, we have that
		\begin{align}
		    P_\text{success}&=\norm{CA^{-1}_{\epsilon}\ket{b}}^2\geq \left(\frac{\norm{\ket{\bm{b}}}C}{\norm{A_\epsilon}}\right)^2 \\	    &\geq\left(\frac{(1-\epsilon_R\lambda_{\min} /C)C}{\lambda_{\max}}\right)^2.
		\end{align}
    Therefore, to achieve an overall error $\epsilon$, we define $C=\lambda_{\min}$, which means that we will rotate by $C^\prime = \frac{N_l t\lambda_{\min}}{2\pi}$, and take
    $\epsilon_R = \frac{\epsilon}{2\kappa^2-\epsilon}$, $\epsilon_C = \epsilon_R/2$, and
    \begin{equation}
        n_1 = n_2 = n_3 = \left\lfloor \log(\frac{2(2\kappa^2 - \epsilon)}{\epsilon}+1)\right\rfloor+1.
    \end{equation}
	This concludes the proof.	
	\end{proof}
	
	\section{Technical proofs from Sec.\ref{obsersec}}\label{Appendix:observables}
	\subsection{Proof of Proposition~\ref{normobv}}\label{Sec:normobv}
	\begin{proof}
	With the notation from Prop.~\ref{normobv},
	\begin{equation}
	    \bra{\psi_1}M_{1}^{\dagger}M_{1}\ket{\psi_1} =
		\frac{C^2}{\norm{\bm{b}}^2}\sum_{j=0}^{N-1}\left(\frac{\beta_j}{\tilde{\lambda}_j}\right)^2
		= \frac{C^2\norm{\bm{x}}^{2}}{\norm{\bm{b}}^2},
	\end{equation}
	and
	\begin{equation}
		\bra{\psi_0}M_{1}^{\dagger}M_{1}\ket{\psi_0} =
		\frac{1}{N}\sum_{j=0}^{N-1}\left(\frac{cC\beta_j}{\norm{\bm{b}}_\infty\tilde{\lambda}_j}\right)^2
		=\frac{c^{2}C^2\norm{\bm{x}}^{2}}{N\norm{\bm{b}}_\infty^2}.
	\end{equation}
	\end{proof}
	
	\subsection{Proof of Proposition~\ref{eigerr}}\label{Sec:eigerr}
	\begin{proof}
	We start with $k=1$. Similarly as before, we have for $i\in\{0,1\}$,
	\begin{align}
	    &M_1(i)\ket{\phi_{1,1}}=\\
	    & I_{n_b}\otimes\ket{i}\bra{i}\left(\frac{1}{\sqrt{2}}\sum_{k,l=0}^{1}(-1)^{k\cdot l}\ket{k}\bra{l}\right)\frac{C}{\norm{\bm{b}}}\sum_{j=0}^{N-1}x_j\ket{1}\ket{j}_{n_b}\\
	    &=\frac{C}{\sqrt{2}\norm{\bm{b}}}\sum_{j=0}^{N/2-1}\left((-1)^0x_{2j}+(-1)^ix_{2j+1}\right)\ket{1}\ket{2j+i}_{n_b}.
	\end{align}
	Therefore,
	\begin{align}
	    &\bra{\phi_{1,1}}M^\dagger_1(i)M_1(i)\ket{\phi_{1,1}}
	    =\\
	    &\frac{C^2}{2\norm{\bm{b}}^2}\left(\sum_{j=0}^{N-1}x_j^2+(-1)^i 2\sum_{j=0}^{N/2-1}x_{2j}x_{2j+1}\right),
	\end{align}
	and the result follows.
	
	The proof for $M_1(i)\ket{\phi_{0,1}}$ is the same but with a $\frac{cC}{\sqrt{N}\norm{\bm{b}}_\infty}$ factor instead of $\frac{C}{\norm{\bm{b}}}$.

	Now we will prove the case for $k>1$. 
	Note that from the circuit in Fig.~\ref{hadik}, the outcome $\ket{q_{0}}=\cdots=\ket{q_{k-2}}=\ket{1}$ can only happen when originally either
	\begin{equation}
	\ket{q_{0}}=\cdots=\ket{q_{k-2}}=\ket{1} \quad\text{and}\quad\ket{q_{k-1}}=\ket{0},
	\end{equation}
	or
	\begin{equation}
	\ket{q_{0}}=\cdots=\ket{q_{k-2}}=\ket{0} \quad\text{and}\quad\ket{q_{k-1}}=\ket{1}.
	\end{equation}
	That is, for those basis states $\ket{i}_{n_b}$ such that $i\equiv 2^{k-1}\ (\textrm{mod}\ 2^{k})$ or $i\equiv 2^{k-1}-1\ (\textrm{mod}\ 2^{k})$. 
	Assuming we have already measured $\ket{q_{0}}=\cdots=\ket{q_{k-2}}=\ket{1}$, calculating the probability of $\ket{q_{k-1}}=\ket{0}$ or $\ket{q_{k-1}}=\ket{1}$, the proof is similar as the case for $k=1$ shown above. 
	Therefore,
	\begin{align}
	&\frac{2\norm{\bm{b}}^2}{C^2}\bra{\phi_{1,k}}M^\dagger_k(i)M_k(i)\ket{\phi_{1,k}}
	=\sum^{N-1}_{\substack{j=0\\j\equiv 2^{k-1} \text{ or }2^{k-1} -1\ (\textrm{mod}\ 2^{k})}}x_j^2\\
	&+(-1)^i 2\sum_{\substack{j=0\\j\equiv 2^{k-1} -1\ (\textrm{mod}\ 2^{k})}}^{N-1}x_{j}x_{j+1}
	\end{align}
	and the result follows.
	
	Again, the proof in the case we run state preparation and HHL together is the same changing the $\frac{2\norm{\bm{b}}^2}{C^2}$ factor by $\frac{2N\norm{\bm{b}}_\infty^2}{c^2C^2}$.
	\end{proof}
	
	\subsection{Proof of Proposition~\ref{Prop:obv2}}\label{Sec:Prop:obv2}
	\begin{proof}
	Let $\bullet_{n}$ denote the bitwise dot product. Then we can write	
    \begin{equation}
    H^{\otimes n}=\frac{1}{\sqrt{2^{n}}}\sum_{i,j=0}^{2^{n}-1}(-1)^{i\bullet_{n}j}\ket{i}_{n}\bra{j}_n.
    \end{equation}
    Therefore, writing
    \begin{equation}
         \sum_{j=0}^{N-1}\frac{\beta_j}{\tilde{\lambda}_j}\ket{u_j}_{n_b} = \sum_{i=0}^{N-1} x_i\ket{i}_{n_b},
    \end{equation}
    we have
    \begin{align}
		M_{1,0}\ket{\phi_1}&=\frac{C}{\sqrt{N}\norm{\bm{b}}}\sum_{i=0}^{N-1} x_i\sum_{k=0}^{N-1}(-1)^{k\bullet_{n_b}i}\ket{0}_{n_b}\bra{0}_{n_b}\ket{k}_{n_b}\\
		&= \frac{C}{\sqrt{N}\norm{\bm{b}}}\sum_{i=0}^{N-1} x_i\ket{0}_{n_b}.
	\end{align}
	The proof for $M_{1,0}\ket{\phi_0}$ is the same changing the $C/\norm{\bm{b}}$ factor by $\frac{cC}{\sqrt{N}\norm{\bm{b}}}_\infty$.
	\end{proof}
	
	\section{Properties of extrapolation}\label{Appendix:extrap_properties}
	In this Appendix we show that the extrapolation scheme can be applied to polynomials in the desired result.
	More specifically, if we want to approximate a quantity $A$ as $A=\sum_{j=1}^l a_j A_j$, $p(A)$ can be approximated via $\sum_{j=1}^l a_j p(A_j)$, where $p$ denotes a polynomial of degree $d\in\mathbb{N}$.
	 
	Let $E_j$, $1\leq j\leq l$ denote the error in approximating a quantity $E$ using a time step $t/m_j$, where $t\in\mathbb{R}$ and $m_j,l\in\mathbb{N}$.
	Furthermore, suppose we can write 
	\begin{equation}
	    E_j = \sum_{k=a}^\infty \alpha_k \left(\frac{t}{m_j}\right)^k\quad\text{where}\quad \alpha_k\in \mathbb{C}^{n\times n}, n\in\mathbb{N}.
	\end{equation}
	Let $a_j\in\mathbb{R}$,$1\leq j \leq l$, be the coefficients of an extrapolation method that yields an $l^\prime$ order approximation, i.e.
	\begin{equation}
	    \sum_{j=1}^l a_j E_j = \mathcal{O}(t^{l^\prime+1}),
	\end{equation}
	or, equivalently, for $k=a,\dots,l^\prime$,
	\begin{equation}
	    \sum_{j=1}^l a_j\alpha_k \left(\frac{t}{m_j}\right)^k = 
	    \alpha_k t^k\sum_{j=1}^l a_j \left(\frac{1}{m_j}\right)^k=0.
	\end{equation}
	
	We will now show that extrapolation with the given $a_j$ coefficients yields an at least $l^\prime$ order approximation for each of the following identities, hence also for any linear combination of them. 
	The intuition behind this result is that extrapolation and the $a_j$'s are independent of the coefficients of the power series, $\alpha_k$, thus the extrapolation works as long as the resulting power series preserves the same structure in terms of $(t/m_j)^k$.
	
	Let $U\in\mathbb{C}^{n\times n}$, then the extrapolation scheme works on any of the following identities:
	\begin{enumerate}
	    \item $UE_j$
	    \item $E_j U$
	    \item $E_j ^ m$, $m\in\mathbb{N}$
	    \item $E_j U E_j$
	\end{enumerate}
	Identities (1)-(2) can be easily deduced by linearity.
	
	Therefore we will proceed to prove identity (3) for $m=2$, and by induction the result holds for larger $m$.
	\begin{proof}
	For $m=2$ we have
	\begin{equation}
	    E_j^2 = \sum_{k=2a}^\infty\sum_{k_1 + k_2 = k}\alpha_{k_1}\alpha_{k_2}\left(\frac{t}{m_j}\right)^k.
	\end{equation}
	Then
	\begin{align}
	    &\sum_{j=1}^l a_j \sum_{k_1 + k_2 = k}\alpha_{k_1}\alpha_{k_2}\left(\frac{t}{m_j}\right)^k \\
	    &=
	    \left(\sum_{k_1 + k_2 = k}\alpha_{k_1}\alpha_{k_2}\right)t^k\sum_{j=1}^l a_j\left(\frac{1}{m_j}\right)^k=0
	\end{align}
	for $k=a,\dots,l^\prime$.
	\end{proof}
	
	Finally it remains to prove identity (4).
	\begin{proof}
	Writing explicitly the power series gives
	\begin{align}
	    E_j U E_j &= \left(\sum_{k=a}^\infty\alpha_k \left(\frac{t}{m_j}\right)^k\right)\sum_{k=a}^\infty U \alpha_k\left(\frac{t}{m_j}\right)^k \\
	    &= \sum_{k=2a}^\infty\sum_{k_1 + k_2 = k}\alpha_{k_1}U\alpha_{k_2}\left(\frac{t}{m_j}\right)^k.
	\end{align}
	Thus, we have that
	\begin{equation}
	    \sum_{j=1}^l a_j\sum_{k_1 + k_2 = k}\alpha_{k_1}U\alpha_{k_2}\left(\frac{t}{m_j}\right)^k
	    = \sum_{k_1 + k_2 = k}\alpha_{k_1}U\alpha_{k_2} t^k\sum_{j=1}^l a_j\left(\frac{1}{m_j}\right)^k = 0
	\end{equation}
	for $k=a,\dots,l^\prime$.
	\end{proof}
	 All points combined means that extrapolation applies to any polynomial in $E+E_j$, therefore the scheme can be applied directly on the value of the observables and not just the solution vector returned by the algorithm.
	
	\section{Real hardware circuit}\label{realhardappendix}
	\begin{figure}[h!]
		\centering
		\includegraphics[width=\linewidth]{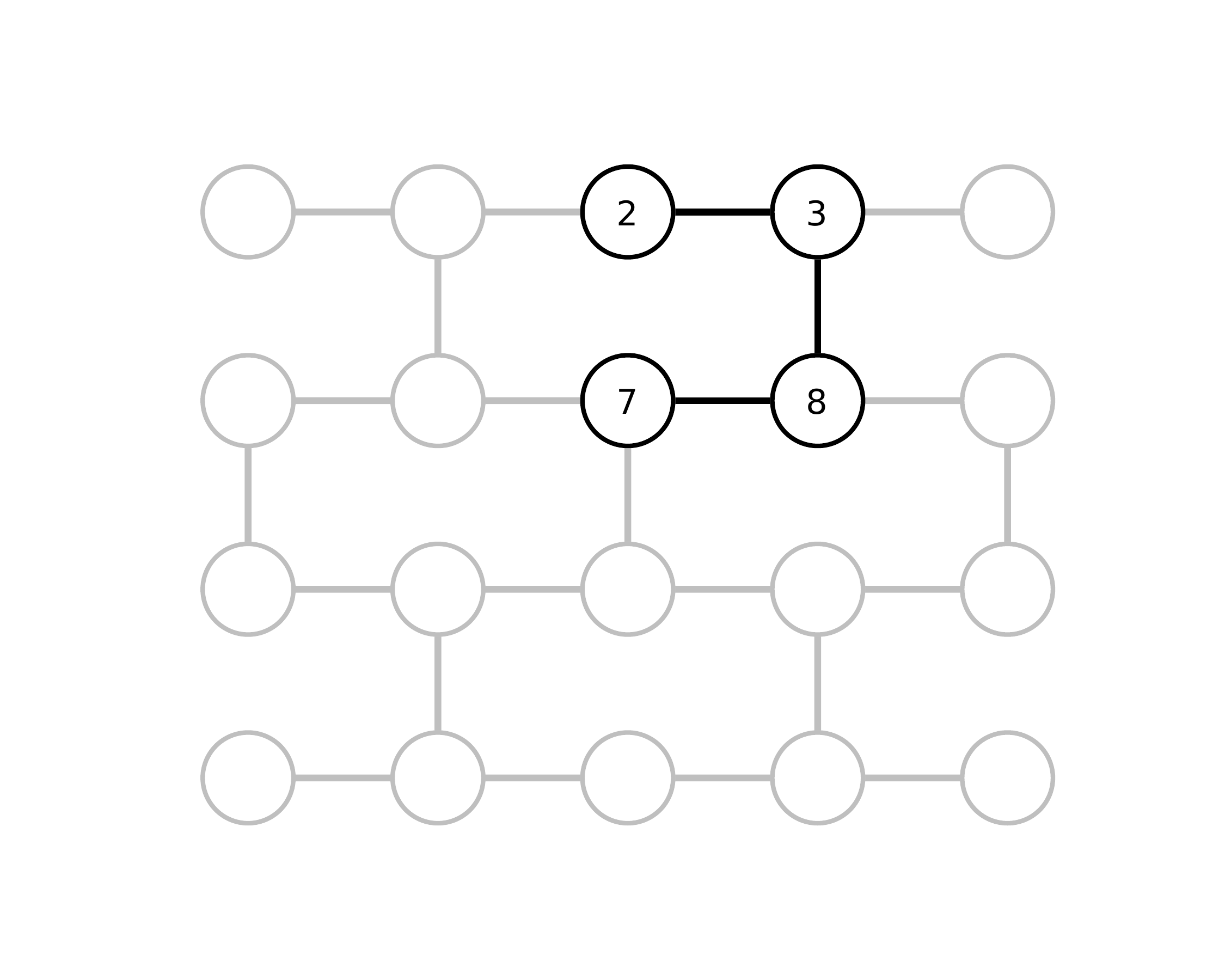}
		
		\caption{Connectivity diagram of the \emph{ibmq$\_\!$boeblingen} $20$-qubit backend provided by IBM Quantum. We use qubits $2,3,7,8$, since these are not fully connected we need additional swaps to run our quantum circuits.}
		\label{connectivity}
	\end{figure}
	
	\begin{figure*}[h!]
		\centering
		\includegraphics[width=\textwidth]{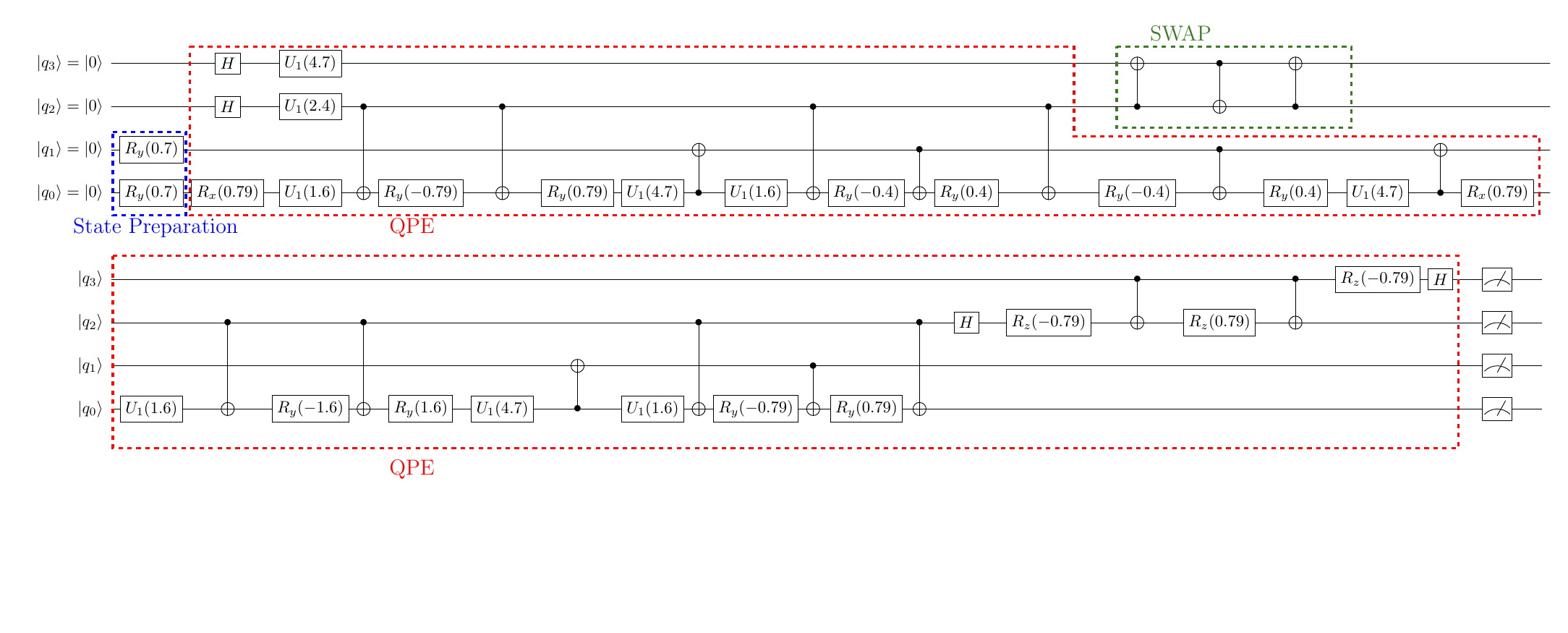}
		
		\caption{HHL circuit for the norm observable without QPE$^{\dagger}$ for $n_{b}=2$, $n_{l}=2$, $b=-1/3$ and $\ket{b}=(\cos^{2}(0.35),\cos(0.35)\sin(0.35),\cos(0.35)\sin(0.35),\sin^{2}(0.35))^{T}$. The conditioned rotation was done classically due to the limitations of the current devices on the width and depth of the circuits that can be run.}
		\label{realcirc1}
	\end{figure*}
	
	\begin{figure*}[h!]
		\centering
		\includegraphics[width=\textwidth]{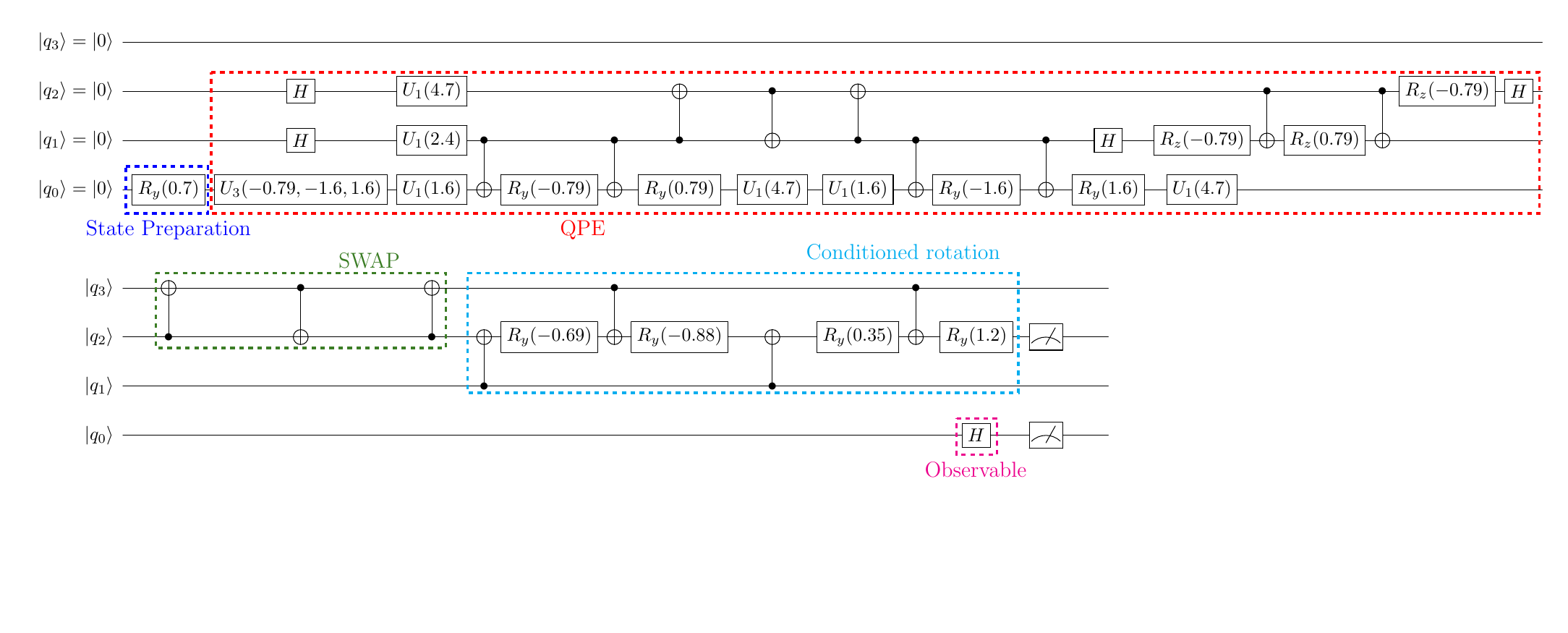}
		
		\caption{HHL circuit for the average observable without QPE$^{\dagger}$ for $n_{b}=1$, $n_{l}=2$, $a=1$, $b=-1/3$ and $\ket{b}=(\cos(0.35),\sin(0.35))^{T}$. The angles for the $R_{y}$ gates in the conditioned rotation part were calculated with the the UniversalQCompiler software \cite{qcompiler}, which allows to optimise a circuit consisting of a sequence of multi-controlled $R_{y}$ gates. }
		\label{realcirc2}
	\end{figure*}
	
\end{document}